\documentclass[journal]{IEEEtran}
\usepackage{booktabs}
\usepackage{amsthm}
\usepackage{ifpdf}
\newtheorem{thm}{Theorem}

\theoremstyle{remark}

\usepackage{algorithm}
\usepackage{algpseudocode}
\usepackage{xcolor}

\usepackage{cite}
\usepackage{multirow}
\usepackage{graphicx}
\usepackage{epstopdf}
\usepackage{lipsum}
\usepackage{longtable}
\usepackage{supertabular}
\usepackage[cmex10]{amsmath}
\usepackage{array}
\ifCLASSOPTIONcompsoc
\else
\usepackage[caption=false]{subfig}
\fi
\usepackage{amsfonts}
\usepackage{float}
\usepackage{psfrag}
\usepackage{enumitem}
\newcounter{MYtempeqncnt}

\begin{document}
%
\title{Absorbing Markov Chain-Based Analysis of Age of Information in Discrete-Time Dual-Queue Systems}

\author{{Yifan Feng, \IEEEmembership{Graduate Student Member, IEEE,} Nail Akar, \IEEEmembership{Senior Member, IEEE,}\\ Zhengchuan Chen, \IEEEmembership{Senior Member, IEEE,} Mehul Motani, \IEEEmembership{Fellow, IEEE}
}

\thanks{Yifan Feng is with the School of Microelectronics and Communication Engineering, Chongqing University, Chongqing 400044, China, and also with the Department of Electrical and Computer Engineering, National University of Singapore, Singapore 117583 (e-mail: fengyf@cqu.edu.cn).}
\thanks{Nail Akar is with the Electrical and Electronics Engineering Department, Bilkent University, Bilkent 06800, Ankara, Turkey (e-mail: akar@ee.bilkent.edu.tr).}
\thanks{Zhengchuan Chen is with the School of Microelectronics and Communication Engineering, Chongqing University, Chongqing 400044, China (e-mail: czc@cqu.edu.cn).}
\thanks{Mehul Motani is with the Department of Electrical and Computer Engineering, the Institute of Data Science, the N.1 Institute for Health, and the Institute for Digital Medicine,
National University of Singapore, Singapore 117583 (e-mail: motani@nus.edu.sg).}
}

\maketitle

\begin{abstract}
Status update systems require the timely collection of sensing information for which deploying multiple sensors/servers to obtain diversity gains is considered as a promising solution.
In this work, we construct an absorbing Markov chain (AMC) to exactly model Age of Information (AoI) in a discrete-time dual-queue (DTDQ) status update system with generate at will (GAW) status updates, discrete phase-type (DPH-type) distributed service times and transmission freezing.
Specifically, transmission is frozen for a certain number of slots following the initiation of a transmission, after which one of the two servers is allowed to simultaneously sample the monitored physical process and transmit a status update packet, according to the availabilities and priorities of the two servers.
Based on the discrete-time AMC, we provide the exact distributions of both AoI and peak AoI (PAoI), enabling the derivation of arbitrary order moments.
In addition, we analytically study the role of freezing using several typical service time distributions, including geometric, uniform, and triangular distributions.
The introduction of freezing for DTDQ systems is demonstrated to be significantly beneficial in reducing the mean AoI for various service time distributions.
Additionally, we study the impact of the statistical parameters of the service times and heterogeneity between the two servers on the freezing gain, i.e., reduction in mean AoI attained with optimum freezing policies.
\end{abstract}


\section{Introduction}
Driven by the demand for convenience and intelligence in vertical industries and daily life, Internet of Things (IoT) applications such as industrial automation, smart homes, and autonomous driving have attracted widespread attention \cite{Survey1,Survey2,Survey3}.
In particular, efficient operation of intelligent devices depends on real-time sensing of relevant environmental and event information \cite{Sensing}. By modeling the acquisition of required information as sampling and transmission of physical processes, the analysis and enhancement of information timeliness for status update systems has become a promising research field in recent years \cite{AOI1}.

In a status update system, focusing on the freshness of the information monitored at the receiving side is more beneficial than focusing on end-to-end latency, for the purpose of ensuring timeliness of data packets for application tasks (such as deciding the next action of a device).
Age of information (AoI) is the most widely known and adopted metric for measuring freshness, defined as the time elapsed since the generation instant of the freshest packet received by the monitor/controller \cite{OrignalAoI}.
Peak AoI (PAoI) is the maximum value reached by AoI before a successful update, which can be essential if the information freshness is required to stay below a given threshold at all times \cite{PAOI1}.
Some variants of AoI, e.g., age of incorrect information (AoII) and value of information (VoI), further consider the dynamic changes in the monitored physical process, aiming to characterize the significance of status updates, in addition to timeliness \cite{AOII,VOI1,VOI2}. AoII- and VoI-related works generally use Markov processes to model the physical processes that serve as information sources.
The authors of \cite{AOII} define AoII as the time elapsed since the source state known to the monitor becomes inconsistent with the actual source state.
For VoI, there are currently at least two definitions.
In \cite{VOI1}, the mutual information between the unobserved real-time state of the source and a number of measurements known to the monitor is termed as VoI,
whereas in \cite{VOI2}, VoI is quantified as the difference between the benefit gained by the monitor receiving a status update packet and the cost of the sensor transmitting that packet.
For information sources modeled by Wiener and Ornstein-Uhlenbeck (OU) processes, mean square error (MSE) is often used as a metric to characterize the significance of state updates, which can be expressed as a function of AoI \cite{MSE1,MSE2,MSE3}.
With the ability to assess the contribution of information to application tasks, AoI and its variants are proposed to be used for goal-oriented semantic communications \cite{AgeIsSem1,AgeIsSem2,AgeIsSem3}.
In particular, despite the existence of many variations of AoI with different goals and complexities, the original AoI as defined in \cite{OrignalAoI} still remains highly relevant and attractive for researchers due to its simplicity, universality, and its insensitivity to the dynamics of the source process.

There are numerous publications in the existing literature focusing on AoI, in which sampling and transmission of the monitored process are modeled as status update packets arriving at the server and serving of these packets, respectively.
Early studies mostly investigated the average AoI and average PAoI of single-source single-server systems, considering different packet arrival distributions, service time distributions, and queuing disciplines.
The authors of \cite{FCFS} derived the average AoI of the M/M/1, M/D/1, and D/M/1 queues under the first-come first-served (FCFS) discipline and further presented the average AoI of the M/M/1/1 queue under the last-come first-served (LCFS) discipline in \cite{LCFS}.
It was shown that the use of LCFS can significantly reduce the average AoI, indicating the advantage of bufferless schemes over infinite buffer schemes in terms of timeliness.
By allowing a buffered packet to be replaced by a newly arriving packet, the authors of \cite{Replace} studied the average AoI of the M/M/1/2 queue and a variant called the M/M/1/2$^*$ queue.
For arbitrary packet arrivals and service time distributions, the AoI distribution of the G/G/1 queue is provided in \cite{General1}, while a general expression of average AoI is derived for the G/G/1/1 queue in \cite{General2}.
The authors of \cite{PH1} obtained the exact distributions of AoI and PAoI for the PH/PH/1/1 and M/PH/1/2 queues through a numerical algorithm with the use of phase-type (PH-type) distributions that can approximate arbitrary distributions accurately, due to their denseness property \cite{neuts81},\cite{ocinneide}. 

Compared with single-source single-server systems, multi-source single-server and single-source multi-server systems are more complex and challenging to analyze, with the former offering device multiplexing benefits, and the latter achieving diversity gains.
Based on stochastic hybrid systems (SHS), the authors of \cite{SHS1} provided a method to derive the average AoI for multi-source single-server systems.
The key point of the SHS method is to find the stationary distribution of a finite-state Markov chain.
In a system consisting of two sources and a server, the SHS method is adopted to obtain the average AoI of each source under three packet management policies \cite{SHS2}.
The SHS method is also introduced for the analysis of multi-source multi-server scenarios by \cite{SHS3}, in which the authors derived a closed-form expression of average AoI for the case consisting of a single source and several homogeneous servers.
The authors of \cite{Mmt1} further extend the SHS method to obtain the moment generating function (MGF) of AoI for a system where a source transmits packets to a monitor via a network.
In \cite{Mmt2}, the MGFs of AoI and PAoI for a multi-source single-server system are derived under both preemptive and non-preemptive policies, based on which the average AoI and PAoI of the two-source case are also provided.
Furthermore, a flow graph-based analysis method is proposed in \cite{Flow1} to derive the average AoI and PAoI of a system with a single source and two servers (i.e., dual-queue system), which has been extended to discrete-time from continuous-time scenarios later \cite{Flow2}.
Recently, the authors of \cite{AMC1} developed a method making use of an absorbing Markov chain (AMC) to obtain the exact distributions of AoI and PAoI for multi-source single-server systems.
Based on the distributions, higher order moments of AoI and PAoI can also be obtained.

Of the aforementioned references, many adopt a random arrival (RA) model, where the generation of status updates is not tied to server availability.
The RA model is suitable when packet generation is sporadic and uncontrollable, but its lack of awareness of server availability leads to waste of resources and queue congestion (in cases of large buffers), especially when the packet arrival rate is high.
A typical example of non-RA models is the zero-wait (ZW) model \cite{Flow1}, where a newly generated status update packet is immediately served once the server completes a transmission.
However, despite its conciseness, the ZW model has been shown to deviate from AoI-optimality for single-queue systems in various cases in \cite{NZW1}, where necessary and sufficient conditions for ZW optimality are provided.
For continuous-time dual-queue systems, the authors of \cite{AMC2} proposed a freezing-based scheme in which the source is frozen for a certain period of time after generating a status update packet.

In this paper, we extend the AMC method to discrete-time scenarios, investigating the distribution of AoI for discrete-time dual-queue (DTDQ) systems when freezing is applied for a deterministic duration of $k$ slots.
The most relevant work is \cite{AMC2}, where a continuous-time AMC is constructed to obtain the distributions of AoI and PAoI under the ZW policy with PH-type service times, and the transmission freezing policy for exponentially distributed service times only. Thus, the analytical modeling of freezing with more general service times and the proper choice of the freezing parameter in such scenarios, are lacking in the literature, which is the scope of the current paper in a discrete-time setting. 
In addition to the inherent distinctions between discrete-time and continuous-time systems, the AMC in this work is more challenging to construct and analyze due to the large size of the corresponding transition matrix of the AMC resulting from the joint consideration of discrete PH-type (DPH-type) distributed service times and freezing.
The main contributions of this work are summarized as follows.
\begin{enumerate}[nosep,wide,labelwidth=!, labelindent=0pt,leftmargin=*]
\item
We construct a discrete-time AMC to characterize the AoI evolution of DTDQ systems with freezing whereas
very general DPH-type distributions are used to model the random service times of the two servers.
To the best of our knowledge, the setting of the paper is more general than the ones in 
the existing literature for dual queue systems. 
This general system model allows us to study the optimum freezing parameter in terms of AoI, the impact of certain statistical parameters of the service times on freezing gain, and the role of heterogeneity between the two servers.
The general system model degenerates to the ZW model if the freezing parameter is set to zero, whereas DPH-type distributions can approximate arbitrary discrete distributions thanks to the denseness property.
\item
Based on the constructed discrete-time AMC, we derive the exact distributions of AoI and PAoI for the considered DTDQ system in matrix geometric form.
By utilizing the distributions, arbitrary order moments of AoI and PAoI can be obtained using matrix-vector operations.
 \item
Using geometrically distributed and uniformly distributed service times as examples, we compare the analytical model with simulations, in terms of average AoI, average squared AoI, and average PAoI.
The strong consistency between the two sets of results demonstrates the validity of the theoretical analysis.	
\item
Through several numerical examples, we study the optimum freezing parameter in terms of mean AoI, and we have shown that for certain parameters of the service times, the use of the optimum freezing parameter can yield a considerable reduction in average AoI.
\end{enumerate}
The rest of this paper is organized as follows.
We describe the system model in Section~\ref{system_model}.
The construction of discrete-time AMC is presented in Section~\ref{AMCconstruct}.
Exact distributions of AoI and PAoI for the investigated DTDQ system are derived in 
Section~\ref{AoIderivation}.
Numerical results are presented in Section~\ref{Results}.
Finally, conclusions are given in Section~\ref{Conclusion}.
\begin{figure}[t]
\centering
\includegraphics[width=0.47\textwidth]{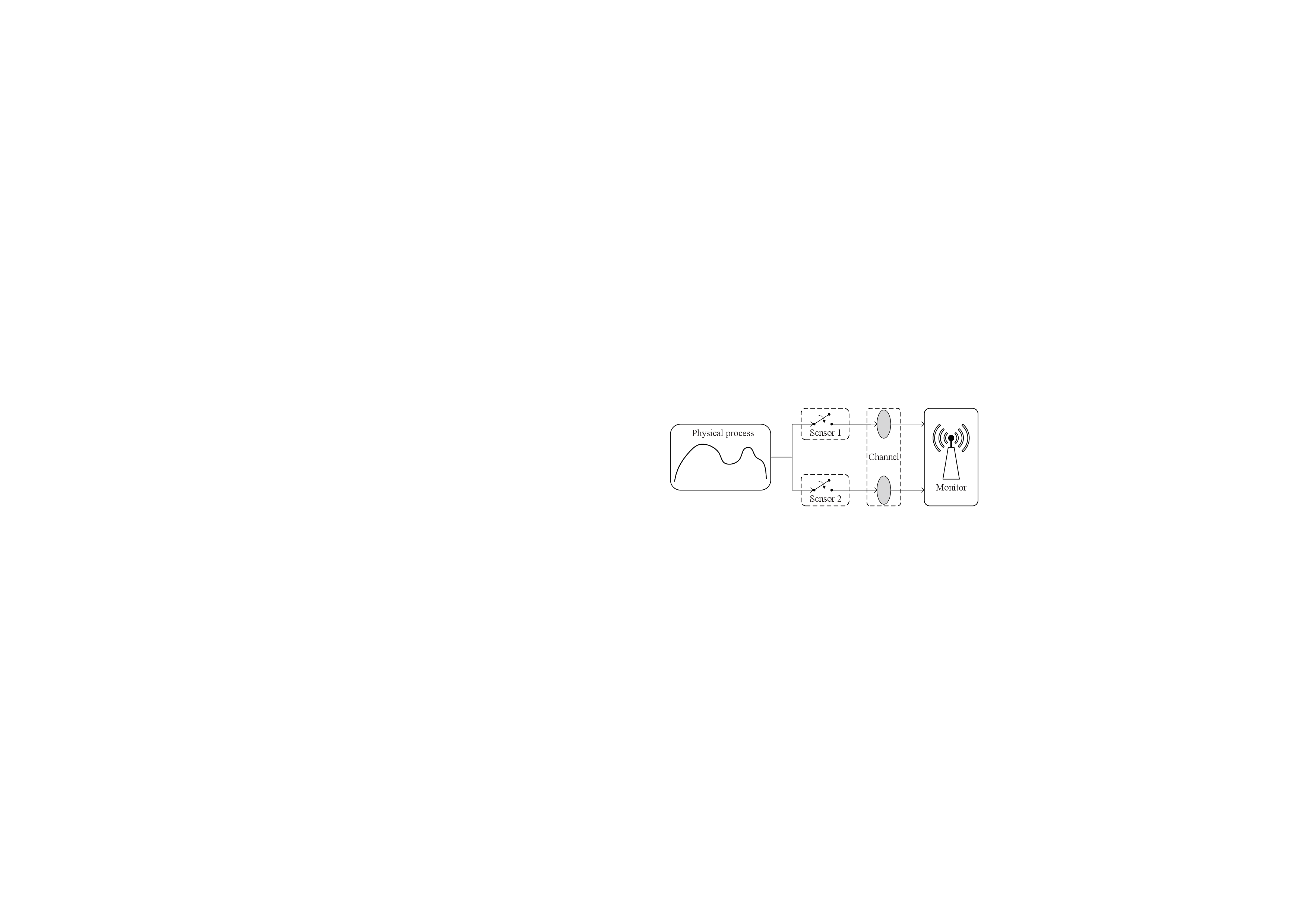}
\caption{DTDQ status update system.}
\label{Model}
\end{figure}

\section{System Model and Preliminaries}\label{system_model}

We consider a non-preemptive DTDQ status update system in which two sensors coordinately sample the same physical process and transmit status update packets to the monitor over the two servers, as shown in Fig.~\ref{Model}.
We follow the generate at will (GAW) model for which the status update system makes a decision on when to sample taking into account a number of factors, including server availability. Due to the randomness of transmission times caused by fading, retransmissions, etc, a packet generated earlier may not reach the monitor before a packet generated later, in which case the former packet becomes obsolete which needs to be discarded by the monitor upon reception
since its reception does not reduce AoI.
A received packet is up to date if its generation was later than the freshest received packet so far.

AoI is used as the main metric in this paper to evaluate the timeliness performance of the DTDQ system.
The evolution of the AoI process $\Delta(t), \ t=0,1,\ldots$ for the established model is illustrated in 
Fig.~\ref{AoI_evolution} for an example scenario.
The time between two consecutive receptions of up-to-date packets amounts to one cycle of the AoI evolution.
The generation and reception instants of the packet that ends Cycle $n$, i.e., the $n$th up-to-date packet, are denoted by $g_n$ and $r_n$, respectively.
In particular, due to the discrete-time setting, the exact end instant of Cycle $n$ is $r_n-1$ instead of $r_n$.
During Cycle $n$, $\Delta(t)$ is incremented by one at each slot, which can be expressed as $\Delta(t)= t - g_{n-1}$.
Note that parallel transmission of two proximately generated packets is more likely to give rise to obsolete packets especially when the servers are homogeneous.
Moreover, it is clear that two proximately generated packets will not be beneficial in improving the timeliness of status updates in terms of AoI even in case neither of the two becomes obsolete.
Consequently, we adopt a sampling strategy in terms of a positive integer freezing parameter $k=1,2,\ldots,$ in order to avoid proximately generated information packets.\footnote{The special case of $k = 0$ (no-freezing) corresponds to the ZW model, where the two servers sample and transmit at any available time slot, and the same packet is transmitted in parallel if both servers are available simultaneously.
The $k = 0$ case is only used as a benchmark for comparisons in the numerical examples.}
Specifically, once a sample is taken and its information packet is transmitted, the next sampling event can be triggered only after $k$ slots, upon the availability of an idle sensor.
When both sensors are busy transmitting status update packets, the sampling needs to be deferred until one of the sensors becomes idle.
When both sensors are idle, one of the sensors is selected to perform sampling and transmission according to a deterministic policy, i.e., either Server 1 ($S_1$) or Server 2 ($S_2$) has strict priority over the other server.
On the other hand, within $k$ slots following a sampling event, no further sampling can be performed even when both sensors become idle. We also note that as the freezing parameter $k$ increases, the servers would not be utilized efficiently leading to a potential increase in AoI. Therefore, the optimum choice of the freezing parameter $k$ 
that will give rise to the most timely system such as one with the lowest mean AoI (or some other AoI-induced metric), is crucial for the DTDQ system of interest.

\begin{figure}[t]
\centering
\includegraphics[width=0.47\textwidth]{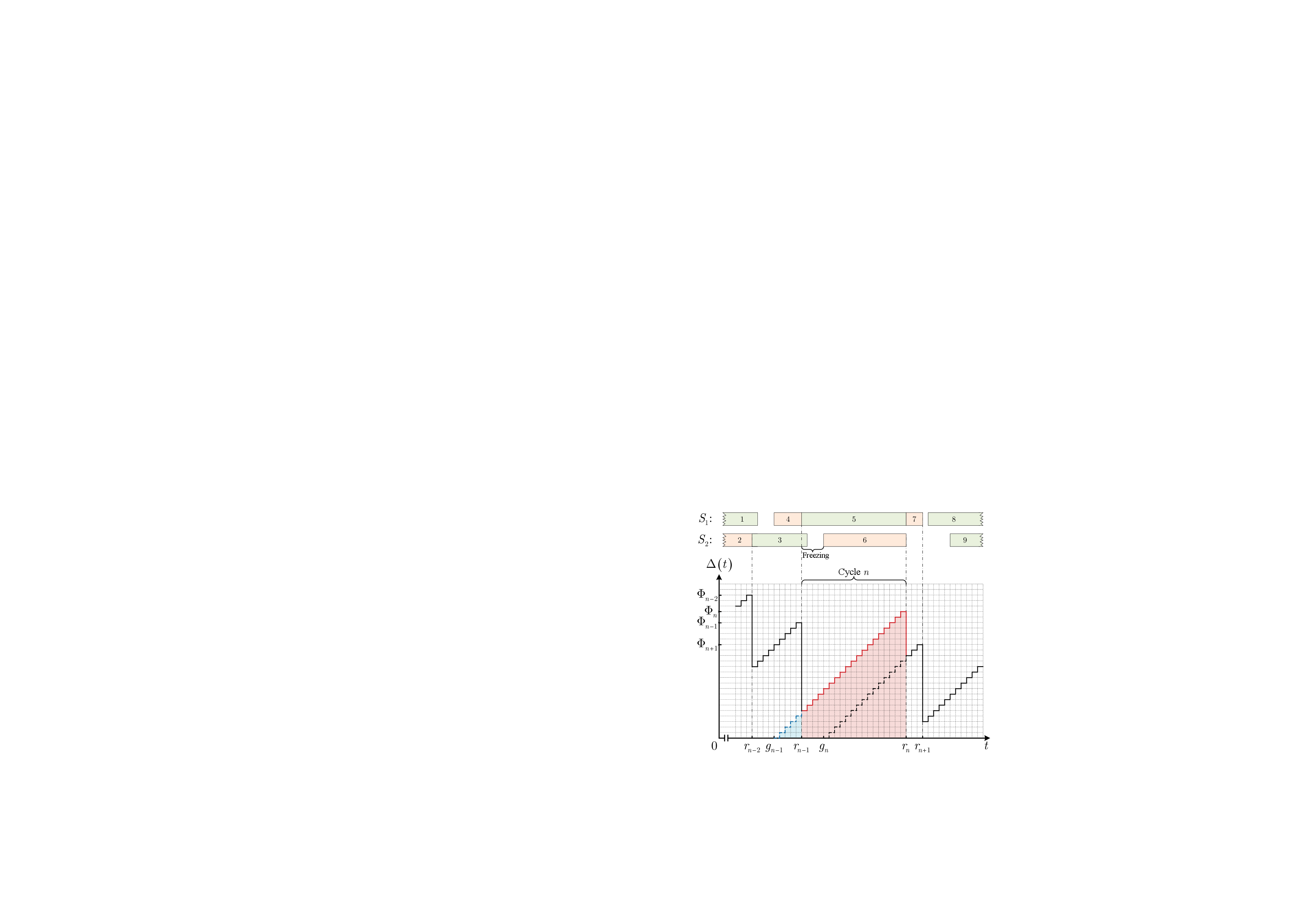}
\caption{An example of the AoI evolution for DTDQ system when $k=4$ and $S_1$ has priority ovr $S_2$.}
\label{AoI_evolution}
\end{figure}

For the example given in Fig.~\ref{AoI_evolution}, the freezing parameter $k$ is set to $4$, and $S_1$ has strict priority over $S_2$ when both servers are idle and a transmission is to take place, and information packets are numbered according to their generation times. 
Note that for this example, Packet 1 and Packet 3 are obsolete since they are received later than Packet 2 and Packet 4, respectively.
When Packet 1 is received, the system is already frozen for one slot, and we wait for $3$ more slots to transmit Packet 4. Similarly, when Packet 3 and Packet 7 are received, we wait for $3$ and $1$ slots, respectively, to comply with the freezing constraint, to transmit the next packet.
When Packets 2, 4, 5, and 6 are received, the freezing durations are already over, hence new packets are generated directly.
Packet 7 and Packet 8 are transmitted over $S_1$ despite the fact that both servers are available due to the priority of $S_1$ over $S_2$.

In this paper, we focus on using the discrete-time AMC-based method to characterize the probability mass function (PMF) of the AoI process $\Delta(t)$ denoted by  $u_{\Delta}(h) = \lim_{t \rightarrow \infty} \Pr (\Delta(t)= h)$ for the steady-state AoI random variable $\Delta$ where $\Pr(\cdot)$ denotes the probability of its argument event. Although the whole PMF will be obtained with our proposed analytical model, 
the mean AoI $\mathbb E\left[ \Delta\right]$ will the main metric of interest in the majority of the numerical examples 
where $\mathbb E\left[ \cdot\right]$ stands for the expectation operator. Note that 
$\mathbb E\left[ \Delta\right]$ is also the average AoI which is $\lim_{K \rightarrow \infty} \frac{1}{K} \sum_{t=0}^{K-1} \Delta(t)$ due to the ergodicity of the AoI process.
Taking Cycle $n$ as the object of observation, the main idea of the AMC method is to construct a discrete-time AMC $\left\{Y_h\right\}, h=0, 1, \ldots$ by considering the generation of the $\left(n-1\right)$st up-to-date packet, denoted by $P^*$, as the initial starting time point and the reception of the $n$th up-to-date packet as the absorption time point.
Clearly, the distribution of steady-state PAoI, denoted by $\Phi$, with PMF $u_{\Phi}(h) = \lim_{n  \rightarrow \infty} \Pr (\Phi_n = h)$, is equal to that of a variable expressed as the absorption time of the AMC $\left\{Y_h\right\}$ minus one.

The main idea of the AMC method \cite{AMC1} is based on the partitioning of the transient states of $\left\{Y_h\right\}$ into two subsets, namely Subset 1 and Subset 2. Subset 1 amounts to the collection of states corresponding to those before the reception of the $\left(n-1\right)$st up-to-date packet 
(blue portion of Fig.~\ref{AoI_evolution}) and these states are called $P^*$ states since in these states $P^*$ is in the system. On the other hand, Subset 2 is the complementary subset of Subset 1 (red portion of Fig.~\ref{AoI_evolution}) and consists of non-$P^*$ states, i.e., states for which $P^*$ is {\em not} in the system. The main idea of the AMC-based approach is that the probability that $\left\{Y_h\right\}$ resides in Subset 2 (or in any of the non-$P^*$ states) at time $h$ turns out to be proportional with  $\Pr (\Delta=h)$.
The detailed exposition of the construction of $\left\{Y_h\right\}$ and the analytical models for the distributions of $\Delta$ and $\Phi$ are presented in Section~\ref{AMCconstruct} and 
Section~\ref{AoIderivation}, respectively. 

Matrices and vectors are denoted by uppercase and lowercase bold letters, respectively. 
Matrix and vector symbols with subscripts are used in this paper to denote the corresponding elements.
A DPH distributed random variable is defined as the time until absorption, denoted by $T$, in a finite-state discrete-time Markov chain (DTMC) $X_h \in \{ 1,\ldots, N, N+1 \}, \ h=0,1,\ldots,$ with the first $N$ states (or phases) being transient states, and the last state being the absorbing state \cite{alfa_book,nielsen_book,telek_book}. 
The DTMC $\{X_h\}$ has the initial probability vector of size $1 \times N$ denoted by $\boldsymbol{\alpha} = \{ \alpha_i \}$ where
$ \alpha_i  = \Pr (X_0 = i), i=1,\ldots,N$,
and probability transition matrix $\boldsymbol G$ of the form,
\begin{align}
\boldsymbol{G} & = \left[ \begin{array}{c|c} \boldsymbol{B} & \boldsymbol{b} \\ \hline  \boldsymbol{0} & 1 \end{array} \right], \label{absorbing}
\end{align}
for a sub-stochastic matrix $\boldsymbol{B}$ and a column vector $\boldsymbol{b}=(\boldsymbol{I}-\boldsymbol{B}) \boldsymbol{1}$ where $\boldsymbol{I}$ denotes the identity matrix. 
$\boldsymbol 0$ and $\boldsymbol{1}$ are an all-zeros row vector and an all-ones column vector of appropriate size, respectively.
In this case, we say $T \sim \text{DPH}(\boldsymbol{\alpha},\boldsymbol{B})$ with order $N$.
The geometric distribution with parameter $p$ is a special case of DPH where $T \sim \text{DPH}(1,1-p)$. 
Discrete distributions with finite support are also DPH-type distributions. 
We refer the reader to \cite{alfa_book} for various DPH-type distributions used in practice. 
Due to the denseness property, service times of arbitrary discrete distributions can be accurately approximated by DPH-type distributions \cite{alfa_book,telek_book}.

In particular, to use $T \sim \text{DPH}(\boldsymbol{\alpha},\boldsymbol{B})$ to represent service times with a finite support PMF that satisfies $u_T(w)=0$ for $w > N$, $\boldsymbol{\alpha}$ should be a row vector of zeros except the first element $\alpha_1=1$, while $\boldsymbol{B}$ is a square matrix of zeros except the element in row $w$ and column $w+1$, which is set to \cite{alfa_book},
\begin{align}
\boldsymbol{B}_{w(w+1)} & = 1-\frac{u_T(w)}{\sum_{a=w}^N u_T(a)}, \ w=1,\ldots,N-1. 
\label{DPHrepresentation}
\end{align}
We assume that in this paper, the service times $T_1$ and $T_2$ of the servers $S_1$ and $S_2$ follow DPH-type distributions $T_1 \sim \text{DPH}\left(\boldsymbol{\alpha^{\left(1\right)}},
\boldsymbol{B^{\left(1\right)}}\right)$ and $T_2 \sim \text{DPH}\left(\boldsymbol{\alpha^{\left(2\right)}},
\boldsymbol{B^{\left(2\right)}}\right)$, with orders $N_1$ and $N_2$, respectively.
The vector of probabilities of transitions from each phase to the absorbing state is given by $\boldsymbol{b^{\left(m\right)}} = \left(\boldsymbol{I} - \boldsymbol{B^{\left(m\right)}}\right) \mathbf{1}, \ m=1,2$.

\section{Discrete-Time AMC Construction with Freezing and DPH-Type Service Times}\label{AMCconstruct}
In this section, we study in detail the construction of the discrete-time AMC $\left\{Y_h\right\}$, which is used to model the operation of DTDQ systems considering freezing and DPH-type service times.
Specifically, $\left\{Y_h\right\}$ is initialized by the arrival of an arbitrary packet $P^*$ into the system, and for Cycle $n$ in Fig.~\ref{AoI_evolution} the successful reception of $P^*$ (which cannot be in an obsolete manner) corresponds to the moment $r_{n-1}$.
After $P^*$ is received, $\left\{Y_h\right\}$ continues to evolve until it moves to the successful absorbing state when the next up-to-date packet is received.
The transition of $\left\{Y_h\right\}$ to the successful absorbing state corresponds to the moment $r_n$.
On the other hand, if $P^*$ is obsolete when being received, it will not be exhibited in Fig.~\ref{AoI_evolution} in the form of a falling edge.
Thus, we set another absorption state, i.e., the unsuccessful absorption state, to ensure the absorption of $\left\{Y_h\right\}$ when this event happens.

\begin{table}
\centering
\caption{State space of the discrete-time AMC $\left\{Y_h\right\}$}
\renewcommand{\arraystretch}{1.5}
\label{Table1}
\begin{tabular}{|c|c|}
\hline
State & Description \\
\hline\hline
$\left(1,i,j,l\right)$ & $P^*$ on $S_1$, $\tau_1 > \tau_2$ \\
\hline
$\left(2,i,j,l\right)$ & $P^*$ on $S_1$, $\tau_1 < \tau_2$ \\
\hline
$\left(3,i,j,l\right)$ & $P^*$ on $S_2$, $\tau_1 > \tau_2$ \\
\hline
$\left(4,i,j,l\right)$ & $P^*$ on $S_2$, $\tau_1 < \tau_2$ \\
\hline
$\left(5,i,0,l\right)$ & $P^*$ on $S_1$, $S_2$ idle \\
\hline
$\left(6,0,j,l\right)$ & $P^*$ on $S_2$, $S_1$ idle \\
\hline
$\left(7,0,0,l\right)$ & $S_1$, $S_2$ idle \\
\hline
$\left(8,i,j\right)$ & $P_1$, $P_2$ up to date \\
\hline
$\left(9,i,0,l\right)$ & $P_1$ obsolete, $S_2$ idle \\
\hline
$\left(10,0,j,l\right)$ & $P_2$ obsolete, $S_1$ idle \\
\hline
$\left(11,i,j,l\right)$ & $P_1$ up to date, $P_2$ obsolete \\
\hline
$\left(12,i,j,l\right)$ & $P_2$ up to date, $P_1$ obsolete \\
\hline
$\left(13,i,0,l\right)$ & $P_1$ up to date, $S_2$ idle \\
\hline
$\left(14,0,j,l\right)$ & $P_2$ up to date, $S_1$ idle \\
\hline
$15$ & Successful absorbing state \\
\hline
$16$ & Unsuccessful absorbing state \\
\hline
\end{tabular}
\end{table}

\begin{table*}
\centering
\caption{Transition probabilities of the discrete-time AMC $\left\{Y_h\right\}$}
\renewcommand{\arraystretch}{1.5}
\label{Table2}
\begin{tabular}{|p{1.28cm}<{\centering}|p{2.9cm}<{\centering}|p{1.1cm}<{\centering}|p{1.85cm}<{\centering}|p{1.28cm}<{\centering}|p{2.9cm}<{\centering}|p{1.1cm}<{\centering}|p{1.85cm}<{\centering}|}
\hline
From & \multicolumn{2}{c|}{To} & Probability & From & \multicolumn{2}{c|}{To} & Probability \\
\hline\hline
\multirow{9}*{$\left(1,i,j,l\right)$} & \multicolumn{2}{c|}{$\left(10,0,j^\prime,l+1\right)$ for $l<k-1$} & $\boldsymbol{b}^{\boldsymbol{\left(1\right)}}_i\boldsymbol{B}^{\boldsymbol{\left(2\right)}}_{jj^\prime}$ & \multirow{2}*{$\left(6,0,j,l\right)$} & \multicolumn{2}{c|}{$\left(6, 0, j^\prime, l+1\right)$ for $l<k-1$} & $\boldsymbol{B}_{jj^\prime}^{\boldsymbol{\left(2\right)}}$ \\
\cline{2-4}\cline{6-8}
~ & \multicolumn{2}{c|}{$\left(11,i^\prime, j^\prime, 0\right)$ for $l=k-1$} & $\boldsymbol{b}^{\boldsymbol{\left(1\right)}}_i\boldsymbol{B}^{\boldsymbol{\left(2\right)}}_{jj^\prime}\boldsymbol{\alpha} ^{\boldsymbol{\left(1\right)}}_{i^\prime}$ & ~ & \multicolumn{2}{c|}{$\left(3, i^\prime, j^\prime, 0\right)$ for $l=k-1$} & $\boldsymbol{B}_{jj^\prime}^{\boldsymbol{\left(2\right)}}\boldsymbol{\alpha} _{i^\prime}^{\boldsymbol{\left(1\right)}}$ \\
\cline{2-8}
~ & \multicolumn{2}{c|}{$\left(7, 0, 0, l+1\right)$ for $l<k-1$} & $\boldsymbol{b}_i^{\boldsymbol{\left(1\right)}}\boldsymbol{b}_j^{\boldsymbol{\left(2\right)}}$ & \multirow{3}*{$\left(7,0,0,l\right)$} & \multicolumn{2}{c|}{$\left(7,0,0,l+1\right)$ for $l<k-1$} & $1$ \\
\cline{2-4}\cline{6-8}
~ & $S_1$ priority: $\left(13,i^\prime,0,0\right)$ & for & $\boldsymbol{b}_i^{\boldsymbol{\left(1\right)}}\boldsymbol{b}_j^{\boldsymbol{\left(2\right)}}\boldsymbol{\alpha} _{i^\prime}^{\boldsymbol{\left(1\right)}}$ & ~ & $S_1$ priority: $\left(13,i^\prime,0,0\right)$ & for & $\boldsymbol{\alpha} _{i^\prime}^{\boldsymbol{\left(1\right)}}$ \\
\cline{2-2}\cline{4-4}\cline{6-6}\cline{8-8}
~ & $S_2$ priority: $\left(14,0,j^\prime,0\right)$ & $l=k-1$ & $\boldsymbol{b}_i^{\boldsymbol{\left(1\right)}}\boldsymbol{b}_j^{\boldsymbol{\left(2\right)}}\boldsymbol{\alpha} _{j^\prime}^{\boldsymbol{\left(2\right)}}$ & ~ & $S_2$ priority: $\left(14,0,j^\prime,0\right)$ & $l=k-1$ & $\boldsymbol{b}_j^{\boldsymbol{\left(2\right)}}\boldsymbol{\alpha} _{j^\prime}^{\boldsymbol{\left(2\right)}}$  \\
\cline{2-8}
~ & \multicolumn{2}{c|}{$\left(5, i^\prime, 0, l+1\right)$ for
$l<k-1$} & $\boldsymbol{b}_j^{\boldsymbol{\left(2\right)}}\boldsymbol{B}_{ii^\prime}^{\boldsymbol{\left(1\right)}}$ & \multirow{3}*{$\left(8,i,j\right)$} & \multicolumn{2}{c|}{\multirow{2}*{$15$}} & $\boldsymbol{b}_i^{\boldsymbol{\left(1\right)}} + \boldsymbol{b}_j^{\boldsymbol{\left(2\right)}} - $ \\
\cline{2-4}
~ & \multicolumn{2}{c|}{$\left(2, i^\prime, j^\prime, 0\right)$ for
$l=k-1$} & $\boldsymbol{b}_j^{\boldsymbol{\left(2\right)}}\boldsymbol{B}_{ii^\prime}^{\boldsymbol{\left(1\right)}}\boldsymbol{\alpha} _{j^\prime}^{\boldsymbol{\left(2\right)}}$ & ~ & \multicolumn{2}{c|}{~} & $\boldsymbol{b}_i^{\boldsymbol{\left(1\right)}}\boldsymbol{b}_j^{\boldsymbol{\left(2\right)}}$ \\
\cline{2-4}\cline{6-8}
~ & \multicolumn{2}{c|}{$\left(1, i^\prime, j^\prime, l+1\right)$ for $l<k-1$} & \multirow{2}*{$\boldsymbol{B}_{ii^\prime}^{\boldsymbol{\left(1\right)}}\boldsymbol{B}_{jj^\prime}^{\boldsymbol{\left(2\right)}}$} & ~ & \multicolumn{2}{c|}{$\left(8,i^\prime,j^\prime\right)$} & $\boldsymbol{B}_{ii^\prime}^{\boldsymbol{\left(1\right)}}\boldsymbol{B}_{jj^\prime}^{\boldsymbol{\left(2\right)}}$ \\
\cline{2-3}\cline{5-8}
~ & \multicolumn{2}{c|}{$\left(1, i^\prime, j^\prime, k-1\right)$ for $l=k-1$} & ~ & \multirow{5}*{$\left(9,i,0,l\right)$} & \multicolumn{2}{c|}{$\left(7,0,0,l+1\right)$ for $l<k-1$} & $\boldsymbol{b}_i^{\boldsymbol{\left(1\right)}}$ \\
\cline{1-4}\cline{6-8}
\multirow{5}*{$\left(2,i,j,l\right)$} & \multicolumn{2}{c|}{$\left(14,0,j^\prime,l+1\right)$ for $l<k-1$} & $\boldsymbol{b}_i^{\boldsymbol{\left(1\right)}}\boldsymbol{B}_{jj^\prime}^{\boldsymbol{\left(2\right)}}$ & ~ & $S_1$ priority: $\left(13,i^\prime,0,0\right)$ & for & $\boldsymbol{b}_i^{\boldsymbol{\left(1\right)}}\boldsymbol{\alpha} _{i^\prime}^{\boldsymbol{\left(1\right)}}$ \\
\cline{2-4}\cline{6-6}\cline{8-8}
~ & \multicolumn{2}{c|}{$\left(8,i^\prime, j^\prime\right)$ for $l=k-1$} & $\boldsymbol{b}_i^{\boldsymbol{\left(1\right)}}\boldsymbol{B}_{jj^\prime}^{\boldsymbol{\left(2\right)}}\boldsymbol{\alpha} _{i^\prime}^{\boldsymbol{\left(1\right)}}$ & ~ & $S_2$ priority: $\left(14,0,j^\prime,0\right)$ & $l=k-1$ & $\boldsymbol{b}_i^{\boldsymbol{\left(1\right)}}\boldsymbol{\alpha} _{j^\prime}^{\boldsymbol{\left(2\right)}}$  \\
\cline{2-4}\cline{6-8}
~ & \multicolumn{2}{c|}{$\left(2, i^\prime, j^\prime, l+1\right)$ for $l<k-1$} & \multirow{2}*{$\boldsymbol{B}_{ii^\prime}^{\boldsymbol{\left(1\right)}}\boldsymbol{B}_{jj^\prime}^{\boldsymbol{\left(2\right)}}$} & ~ & \multicolumn{2}{c|}{$\left(9, i^\prime, 0, l+1\right)$ for $l<k-1$} & $\boldsymbol{B}_{ii^\prime}^{\boldsymbol{\left(1\right)}}$ \\
\cline{2-3}\cline{6-8}
~ & \multicolumn{2}{c|}{$\left(2, i^\prime, j^\prime, k-1\right)$ for $l=k-1$} & ~ & ~ & \multicolumn{2}{c|}{$\left(12, i^\prime, j^\prime, 0\right)$ for $l=k-1$} & $\boldsymbol{B}_{ii^\prime}^{\boldsymbol{\left(1\right)}}\boldsymbol{\alpha} _{j^\prime}^{\boldsymbol{\left(2\right)}}$ \\
\cline{2-8}
~ & \multicolumn{2}{c|}{$16$} & $\boldsymbol{b}_j^{\boldsymbol{\left(2\right)}}$ & \multirow{5}*{$\left(10, 0, j, l\right)$} & \multicolumn{2}{c|}{$\left(7,0,0,l+1\right)$ for $l<k-1$} & $\boldsymbol{b}_j^{\boldsymbol{\left(2\right)}}$ \\
\cline{1-4}\cline{6-8}
\multirow{5}*{$\left(3,i,j,l\right)$} & \multicolumn{2}{c|}{$\left(13,i^\prime,0,l+1\right)$ for $l<k-1$} & $\boldsymbol{b}_j^{\boldsymbol{\left(2\right)}}\boldsymbol{B}_{ii^\prime}^{\boldsymbol{\left(1\right)}}$ & ~ & $S_1$ priority: $\left(13,i^\prime,0,0\right)$ & for & $\boldsymbol{b}_j^{\boldsymbol{\left(2\right)}}\boldsymbol{\alpha} _{i^\prime}^{\boldsymbol{\left(1\right)}}$ \\
\cline{2-4}\cline{6-6}\cline{8-8}
~ & \multicolumn{2}{c|}{$\left(8,i^\prime, j^\prime\right)$ for $l=k-1$} & $\boldsymbol{b}_j^{\boldsymbol{\left(2\right)}}\boldsymbol{B}_{ii^\prime}^{\boldsymbol{\left(1\right)}}\boldsymbol{\alpha} _{j^\prime}^{\boldsymbol{\left(2\right)}}$ & ~ & $S_2$ priority: $\left(14,0,j^\prime,0\right)$ & $l=k-1$ & $\boldsymbol{b}_i^{\boldsymbol{\left(1\right)}}\boldsymbol{\alpha} _{j^\prime}^{\boldsymbol{\left(2\right)}}$ \\
\cline{2-4}\cline{6-8}
~ & \multicolumn{2}{c|}{$\left(3, i^\prime, j^\prime, l+1\right)$ for $l<k-1$} & \multirow{2}*{$\boldsymbol{B}_{ii^\prime}^{\boldsymbol{\left(1\right)}}\boldsymbol{B}_{jj^\prime}^{\boldsymbol{\left(2\right)}}$} & ~ & \multicolumn{2}{c|}{$\left(10, 0, j^\prime, l+1\right)$ for $l<k-1$} &
$\boldsymbol{B}_{jj^\prime}^{\boldsymbol{\left(2\right)}}$ \\
\cline{2-3}\cline{6-8}
~ & \multicolumn{2}{c|}{$\left(3, i^\prime, j^\prime, k-1\right)$ for $l=k-1$} & ~ & ~ & \multicolumn{2}{c|}{$\left(11, i^\prime, j^\prime, 0\right)$ for $l=k-1$} &
$\boldsymbol{B}_{jj^\prime}^{\boldsymbol{\left(2\right)}}\boldsymbol{\alpha} _{i^\prime}^{\boldsymbol{\left(1\right)}}$ \\
\cline{2-8}
~ & \multicolumn{2}{c|}{$16$} & $\boldsymbol{b}_i^{\boldsymbol{\left(1\right)}}$ & \multirow{5}*{$\left(11, i, j, l\right)$} & \multicolumn{2}{c|}{$15$} & $\boldsymbol{b}_i^{\boldsymbol{\left(1\right)}}$ \\
\cline{1-4}\cline{6-8}
\multirow{9}*{$\left(4,i,j,l\right)$} & \multicolumn{2}{c|}{$\left(6,0,j^\prime,l+1\right)$ for $l<k-1$} & $\boldsymbol{b}_i^{\boldsymbol{\left(1\right)}}\boldsymbol{B}_{jj^\prime}^{\boldsymbol{\left(2\right)}}$ & ~ & \multicolumn{2}{c|}{$\left(13,i^\prime,0,l+1\right)$ for $l<k-1$} & $\boldsymbol{b}_j^{\boldsymbol{\left(2\right)}}\boldsymbol{B}_{ii^\prime}^{\boldsymbol{\left(1\right)}}$ \\
\cline{2-4}\cline{6-8}
~ & \multicolumn{2}{c|}{$\left(3,i^\prime, j^\prime, 0\right)$ for $l=k-1$} & $\boldsymbol{b}_i^{\boldsymbol{\left(1\right)}}\boldsymbol{B}_{jj^\prime}^{\boldsymbol{\left(2\right)}}\boldsymbol{\alpha} _{i^\prime}^{\boldsymbol{\left(1\right)}}$ & ~ & \multicolumn{2}{c|}{$\left(8,i^\prime, j^\prime\right)$ for $l=k-1$} & $\boldsymbol{b}_j^{\boldsymbol{\left(2\right)}}\boldsymbol{B}_{ii^\prime}^{\boldsymbol{\left(1\right)}}\boldsymbol{\alpha} _{j^\prime}^{\boldsymbol{\left(2\right)}}$ \\
\cline{2-4}\cline{6-8}
~ & \multicolumn{2}{c|}{$\left(7, 0, 0, l+1\right)$ for $l<k-1$} & $\boldsymbol{b}_i^{\boldsymbol{\left(1\right)}}\boldsymbol{b}_j^{\boldsymbol{\left(2\right)}}$ & ~ & \multicolumn{2}{c|}{$\left(11, i^\prime, j^\prime, l+1\right)$ for $l<k-1$} & \multirow{2}*{$\boldsymbol{B}_{ii^\prime}^{\boldsymbol{\left(1\right)}}\boldsymbol{B}_{jj^\prime}^{\boldsymbol{\left(2\right)}}$}  \\
\cline{2-4}\cline{6-7}
~ & $S_1$ priority: $\left(13,i^\prime,0,0\right)$ & for & $\boldsymbol{b}_i^{\boldsymbol{\left(1\right)}}\boldsymbol{b}_j^{\boldsymbol{\left(2\right)}}\boldsymbol{\alpha} _{i^\prime}^{\boldsymbol{\left(1\right)}}$ & ~ & \multicolumn{2}{c|}{$\left(11, i^\prime, j^\prime, k-1\right)$ for $l=k-1$} & ~ \\
\cline{2-2}\cline{4-8}
~ & $S_2$ priority: $\left(14,0,j^\prime,0\right)$ & $l=k-1$ & $\boldsymbol{b}_i^{\boldsymbol{\left(1\right)}}\boldsymbol{b}_j^{\boldsymbol{\left(2\right)}}\boldsymbol{\alpha} _{j^\prime}^{\boldsymbol{\left(2\right)}}$ & \multirow{5}*{$\left(12, i, j, l\right)$} & \multicolumn{2}{c|}{$15$} & $\boldsymbol{b}_j^{\boldsymbol{\left(2\right)}}$ \\
\cline{2-4}\cline{6-8}
~ & \multicolumn{2}{c|}{$\left(9, i^\prime, 0, l+1\right)$ for
$l<k-1$} & $\boldsymbol{b}_j^{\boldsymbol{\left(2\right)}}\boldsymbol{B}_{ii^\prime}^{\boldsymbol{\left(1\right)}}$ & ~ & \multicolumn{2}{c|}{$\left(14,0,j^\prime,l+1\right)$ for $l<k-1$} & $\boldsymbol{b}_i^{\boldsymbol{\left(1\right)}}\boldsymbol{B}_{jj^\prime}^{\boldsymbol{\left(2\right)}}$ \\
\cline{2-4}\cline{6-8}
~ & \multicolumn{2}{c|}{$\left(12, i^\prime, j^\prime, 0\right)$ for
$l=k-1$} & $\boldsymbol{b}_j^{\boldsymbol{\left(2\right)}}\boldsymbol{B}_{ii^\prime}^{\boldsymbol{\left(1\right)}}\boldsymbol{\alpha} _{j^\prime}^{\boldsymbol{\left(2\right)}}$ & ~ & \multicolumn{2}{c|}{$\left(8,i^\prime, j^\prime\right)$ for $l=k-1$} & $\boldsymbol{b}_i^{\boldsymbol{\left(1\right)}}\boldsymbol{B}_{jj^\prime}^{\boldsymbol{\left(2\right)}}\boldsymbol{\alpha} _{i^\prime}^{\boldsymbol{\left(1\right)}}$ \\
\cline{2-4}\cline{6-8}
~ & \multicolumn{2}{c|}{$\left(4, i^\prime, j^\prime, l+1\right)$ for $l<k-1$} & \multirow{2}*{$\boldsymbol{B}_{ii^\prime}^{\boldsymbol{\left(1\right)}}\boldsymbol{B}_{jj^\prime}^{\boldsymbol{\left(2\right)}}$} & ~ & \multicolumn{2}{c|}{$\left(12, i^\prime, j^\prime, l+1\right)$ for $l<k-1$} & \multirow{2}*{$\boldsymbol{B}_{ii^\prime}^{\boldsymbol{\left(1\right)}}\boldsymbol{B}_{jj^\prime}^{\boldsymbol{\left(2\right)}}$} \\
\cline{2-3}\cline{6-7}
~ & \multicolumn{2}{c|}{$\left(4, i^\prime, j^\prime, k-1\right)$ for $l=k-1$} & ~ & ~ & \multicolumn{2}{c|}{$\left(12, i^\prime, j^\prime, k-1\right)$ for $l=k-1$} & ~  \\
\cline{1-8}
\multirow{5}*{$\left(5,i,0,l\right)$} & \multicolumn{2}{c|}{$\left(7,0,0,l+1\right)$ for $l<k-1$} & $\boldsymbol{b}_i^{\boldsymbol{\left(1\right)}}$ & \multirow{3}*{$\left(13,i,0,l\right)$} & \multicolumn{2}{c|}{$15$} & $\boldsymbol{b}_i^{\boldsymbol{\left(1\right)}}$ \\
\cline{2-4}\cline{6-8}
~ & $S_1$ priority: $\left(13,i^\prime,0,0\right)$ & for & $\boldsymbol{b}_i^{\boldsymbol{\left(1\right)}}\boldsymbol{\alpha} _{i^\prime}^{\boldsymbol{\left(1\right)}}$ & ~ & \multicolumn{2}{c|}{$\left(13,i^\prime,0,l+1\right)$ for $l<k-1$} & $\boldsymbol{B}_{ii^\prime}^{\boldsymbol{\left(1\right)}}$ \\
\cline{2-2}\cline{4-4}\cline{6-8}
~ & $S_2$ priority: $\left(14,0,j^\prime,0\right)$ & $l=k-1$ & $\boldsymbol{b}_i^{\boldsymbol{\left(1\right)}}\boldsymbol{\alpha} _{j^\prime}^{\boldsymbol{\left(2\right)}}$ & ~ & \multicolumn{2}{c|}{$\left(8,i^\prime,j^\prime\right)$ for $l=k-1$} & $\boldsymbol{B}_{ii^\prime}^{\boldsymbol{\left(1\right)}}\boldsymbol{\alpha} _{j^\prime}^{\boldsymbol{\left(2\right)}}$ \\
\cline{2-8}
~ & \multicolumn{2}{c|}{$\left(5, i^\prime, 0, l+1\right)$ for $l<k-1$} & $\boldsymbol{B}_{ii^\prime}^{\boldsymbol{\left(1\right)}}$ & \multirow{3}*{$\left(14,0,j,l\right)$} & \multicolumn{2}{c|}{$15$} & $\boldsymbol{b}_j^{\boldsymbol{\left(2\right)}}$ \\
\cline{2-4}\cline{6-8}
~ & \multicolumn{2}{c|}{$\left(2, i^\prime, j^\prime, 0\right)$ for $l=k-1$} & $\boldsymbol{B}_{ii^\prime}^{\boldsymbol{\left(1\right)}}\boldsymbol{\alpha} _{j^\prime}^{\boldsymbol{\left(2\right)}}$ & ~ & \multicolumn{2}{c|}{$\left(14,0,j^\prime,l+1\right)$ for $l<k-1$} & $\boldsymbol{B}_{jj^\prime}^{\boldsymbol{\left(2\right)}}$ \\
\cline{1-4}\cline{6-8}
\multirow{3}*{$\left(6,0,j,l\right)$} & \multicolumn{2}{c|}{$\left(7,0,0,l+1\right)$ for $l<k-1$} & $\boldsymbol{b}_j^{\boldsymbol{\left(2\right)}}$ & ~ & \multicolumn{2}{c|}{$\left(8,i^\prime,j^\prime\right)$ for $l=k-1$} & $\boldsymbol{B}_{jj^\prime}^{\boldsymbol{\left(2\right)}}\boldsymbol{\alpha} _{i^\prime}^{\boldsymbol{\left(1\right)}}$ \\
\cline{2-8}
~ & $S_1$ priority: $\left(13,i^\prime,0,0\right)$ & for & $\boldsymbol{b}_j^{\boldsymbol{\left(2\right)}}\boldsymbol{\alpha} _{i^\prime}^{\boldsymbol{\left(1\right)}}$ & $15$ & \multicolumn{3}{c|}{Absorbing state and no transition} \\
\cline{2-2}\cline{4-8}
~ & $S_2$ priority: $\left(14,0,j^\prime,0\right)$ & $l=k-1$ & $\boldsymbol{b}_j^{\boldsymbol{\left(2\right)}}\boldsymbol{\alpha} _{j^\prime}^{\boldsymbol{\left(2\right)}}$ & $16$ & \multicolumn{3}{c|}{Absorbing state and no transition} \\
\hline
\end{tabular}
\end{table*}

The state space of $\left\{Y_h\right\}$ is shown in Table \ref{Table1}, including fourteen classes of transient states and two absorbing states.
The first six classes of transient states correspond to the cases when $P^*$ has not yet been received, while the last eight classes correspond to the cases after $P^*$ has been successfully received.
Let $\tau_1$ and $\tau_2$ denote the generation instants of the packets that $S_1$ and $S_2$ are each serving.
$P_1$ and $P_2$ denote the packets served by $S_1$ and $S_2$, respectively, after $P^*$ is successfully received.
Due to the use of freezing and DPH-type service times, the transient states are all four-dimensional, where $i=1,2,\ldots,N_1$ is the phase of service time for $S_1$, $j=1,2,\ldots,N_2$ is the phase of service time for $S_2$, and $l=0,1,\ldots,k-1$ is the freezing clock, representing the time elapsed since a packet was generated.
In particular, when $\left\{Y_h\right\}$ is in $\left(8,i,j\right)$, the reception of either of the two up-to-date packets is capable of bringing $\left\{Y_h\right\}$ to the successful absorbing state, independent of the freezing clock so that $l$ can be omitted.
In addition, for transient states in which $S_1$ or $S_2$ is idle, we write $i$ or $j$ as zero.

Based on the designed state space, we further investigate the transition probability matrix of $\left\{Y_h\right\}$, which is given by
\begin{align}
\label{TransitionMatrix}
\boldsymbol Q = \left[ {\begin{array}{*{20}{c}}
\boldsymbol A&{\boldsymbol{c_s}}&{\boldsymbol{c_u}}\\
{\boldsymbol 0}&1&{0}\\
{\boldsymbol 0}&{0}&1
\end{array}} \right].
\end{align}
$\boldsymbol A$ is an $M \times M$ matrix consisting of transition probabilities among the transient states of $\left\{Y_h\right\}$, where $M$ is the total number of transient states.
$\boldsymbol c_s$ and $\boldsymbol c_u$ are $M \times 1$ column vectors consisting of the probabilities of $\left\{Y_h\right\}$ transitioning from the transient states to the successful absorbing state and the unsuccessful absorbing state, respectively.
Furthermore, $M$ depends on the freezing parameter $k$ and the number of phases $N_1$ and $N_2$ for the two service time distributions.
Specifically, we have
\begin{align}
\label{OurderOfA}
M = k\left(6N_1N_2 + 3\left(N_1 + N_2\right) + 1\right) + N_1N_2.
\end{align}
From (\ref{OurderOfA}), it can be seen that the order of $\boldsymbol A$ can easily be quite large due to the joint consideration of freezing and DPH-type service times.
The minimum value of $M$ is $14$, which occurs when $N_1=N_2=1$ (e.g., geometrically distributed service times) and $k=1$.
If we increase the three parameters, say $N_1=N_2=2$ and $k=3$, $M$ increases to $115$.
The specific transition probabilities of $\left\{Y_h\right\}$ are given in Table \ref{Table2}, where $i^\prime$ and $j^\prime$ denote the new phases of service times for $S_1$ and $S_2$ after a state transition, respectively.
For the cases where both servers are available, we use in Table \ref{Table2} a policy of fixed choice of $S_1$ or $S_2$, i.e., $S_1$ priority or $S_2$ priority, to serve the newly generated packets.

Take the state transitions from $\left(1,i,j,l\right)$ as an example to explain Table \ref{Table2}.
When $\left\{Y_h\right\}$ is in $\left(1,i,j,l\right)$, $P^*$ is on $S_1$ and is more recent than the packet on $S_2$.
There are eight possible events of state transition that can occur in the next step, depending on whether the packets on $S_1$ and $S_2$ are received and whether the freezing clock reaches $k - 1$.
For the cases where $P^*$ is received by the monitor and $P_2$ is not, $S_1$ will become idle if $l<k-1$ and the timing of the freezing clock is increased by $1$, corresponding to $\left(10,0,j^\prime,l+1\right)$, while a newly generated $P_1$ will be serviced by $S_1$ if $l=k - 1$ and the freezing clock returns to $0$, corresponding to $\left(11,i^\prime, j^\prime, 0\right)$.
Note that once a state transition occurs and there is still a packet on a server, the phase of its service time will also transfer, including the possibility of transferring to the same phase.
For the cases where $P^*$ and $P_2$ are both received, both $S_1$ and $S_2$ will become idle when $l<k-1$ and the timing of the freezing clock is increased by $1$, corresponding to $\left(7,0,0,l+1\right)$. When $l=k - 1$, a newly generated $P_1$ or $P_2$ will be serviced by $S_1$ or $S_2$, the other server will become idle, and the freezing clock returns to $0$, corresponding to $\left(13,i^\prime, 0, 0\right)$ or $\left(14,0,j^\prime,0\right)$.
For the cases where $P_2$ is received and $P^*$ is not, the changes in $S_1$ and $S_2$ are contrary to that of the cases where $P^*$ is received and $P_2$ is not.
The next state of $\left\{Y_h\right\}$ is $\left(5, i^\prime, 0, l+1\right)$ if $l<k-1$ and $\left(2, i^\prime, j^\prime, 0\right)$ if $l=k-1$.
In addition, for the cases where neither $P^*$ nor $P_2$ are received by the monitor, the only change in $S_1$ and $S_2$ is a transition in the phases of their service times. No new packet is generated even though the freezing clock has reached $k-1$, which continues to maintain $l = k-1$ in the next step.

In particular, when $\left\{Y_h\right\}$ is in $\left(2,i,j,l\right)$ and $\left(3,i,j,l\right)$, the generation moment of $P^*$ precedes that of another packet in the system, whose receipt is possible no later than $P^*$.
If this event occurs, $P^*$ will become obsolete, thus causing $\left\{Y_h\right\}$ to terminate its evolution in a way that transfers to one of the absorbing states, $16$.
On the other hand, for $\left(8,i,j\right)$, $\left(9,i,0,l\right)$, $\left(10,0,j,l\right)$, $\left(11,i,j,l\right)$, $\left(12,i,j,l\right)$, $\left(13,i,0,l\right)$, and $\left(14,0,j,l\right)$, $\left\{Y_h\right\}$ can transfer to the other absorbing state, $15$, once a up-to-date packet is received by the monitor.
It can be observed as a falling edge in Fig.~\ref{AoI_evolution} only if $\left\{Y_h\right\}$ terminates with a transition to the state $15$, which is the situation we mainly focus on.

Based on Table \ref{Table2}, we can easily derive the expressions of $\boldsymbol A$, $\boldsymbol c_s$, $\boldsymbol c_u$, and thus $\boldsymbol Q$.
Given the generally large order and limited space, we give the exact expressions for the case of $M = 14$ as an example.
Assuming the policy of $S_1$ priority, $\boldsymbol A$, $\boldsymbol c_s$, and $\boldsymbol c_u$ are provided by (\ref{QforM14}), (\ref{Cs}), and (\ref{Cu}), respectively.
The symbol $\top$ denotes the transpose operation.
Since $N_1 = N_2 = 1$ in this case, $\boldsymbol{B^{\left(m\right)}}$, $\boldsymbol{b^{\left(m\right)}}$, and $\boldsymbol{\alpha^{\left(m\right)}}$ all contain only one element and $\boldsymbol{\alpha}^{\boldsymbol{\left(m\right)}}_{1}$ is equal to $1$.
Similarly, we can obtain the transition probability matrix of $\left\{Y_h\right\}$ for any DPH-type service times and freezing parameters, even though the order might be much larger than that of (\ref{QforM14}).

\begin{figure*}[!t]
\normalsize
\setcounter{MYtempeqncnt}{\value{equation}}
\setcounter{equation}{4}
\begin{equation}
\label{QforM14}
\setlength{\arraycolsep}{3.4pt}
\boldsymbol{A}=\left[ {\begin{array}{*{20}{c}}
\boldsymbol{B}_{11}^{\boldsymbol{\left( 1 \right)}}\boldsymbol{B}_{11}^{\boldsymbol{\left( 2 \right)}}&\boldsymbol{b}_1^{\boldsymbol{\left( 2 \right)}}\boldsymbol{B}_{11}^{\boldsymbol{\left( 1 \right)}}&0&0&0&0&0&0&0&0&\boldsymbol{b}_1^{\boldsymbol{\left( 1 \right)}}\boldsymbol{B}_{11}^{\boldsymbol{\left( 2 \right)}}&0&\boldsymbol{b}_1^{\boldsymbol{\left( 1 \right)}}\boldsymbol{b}_1^{\boldsymbol{\left( 2 \right)}}&0\\
0&\boldsymbol{B}_{11}^{\boldsymbol{\left( 1 \right)}}\boldsymbol{B}_{11}^{\boldsymbol{\left( 2 \right)}}&0&0&0&0&0&\boldsymbol{b}_1^{\boldsymbol{\left( 1 \right)}}\boldsymbol{B}_{11}^{\boldsymbol{\left( 2 \right)}}&0&0&0&0&0&0\\
0&0&\boldsymbol{B}_{11}^{\boldsymbol{\left( 1 \right)}}\boldsymbol{B}_{11}^{\boldsymbol{\left( 2 \right)}}&0&0&0&0&\boldsymbol{b}_1^{\boldsymbol{\left( 2 \right)}}\boldsymbol{B}_{11}^{\boldsymbol{\left( 1 \right)}}&0&0&0&0&0&0\\
0&0&\boldsymbol{b}_1^{\boldsymbol{\left( 1 \right)}}\boldsymbol{B}_{11}^{\boldsymbol{\left( 2 \right)}}&\boldsymbol{B}_{11}^{\boldsymbol{\left( 1 \right)}}\boldsymbol{B}_{11}^{\boldsymbol{\left( 2 \right)}}&0&0&0&0&0&0&0&\boldsymbol{b}_1^{\boldsymbol{\left( 2 \right)}}\boldsymbol{B}_{11}^{\boldsymbol{\left( 1 \right)}}&\boldsymbol{b}_1^{\boldsymbol{\left( 1 \right)}}\boldsymbol{b}_1^{\boldsymbol{\left( 2 \right)}}&0\\
0&\boldsymbol{B}_{11}^{\boldsymbol{\left( 1 \right)}}&0&0&0&0&0&0&0&0&0&0&\boldsymbol{b}_1^{\boldsymbol{\left( 1 \right)}}&0\\
0&0&\boldsymbol{B}_{11}^{\boldsymbol{\left( 2 \right)}}&0&0&0&0&0&0&0&0&0&\boldsymbol{b}_1^{\boldsymbol{\left( 2 \right)}}&0\\
0&0&0&0&0&0&0&0&0&0&0&0&1&0\\
0&0&0&0&0&0&0&\boldsymbol{B}_{11}^{\boldsymbol{\left( 1 \right)}}\boldsymbol{B}_{11}^{\boldsymbol{\left( 2 \right)}}&0&0&0&0&0&0\\
0&0&0&0&0&0&0&0&0&0&0&\boldsymbol{B}_{11}^{\boldsymbol{\left( 1 \right)}}&\boldsymbol{b}_1^{\boldsymbol{\left( 1 \right)}}&0\\
0&0&0&0&0&0&0&0&0&0&\boldsymbol{B}_{11}^{\boldsymbol{\left( 2 \right)}}&0&\boldsymbol{b}_1^{\boldsymbol{\left( 2 \right)}}&0\\
0&0&0&0&0&0&0&\boldsymbol{b}_1^{\boldsymbol{\left( 2 \right)}}\boldsymbol{B}_{11}^{\boldsymbol{\left( 1 \right)}}&0&0&\boldsymbol{B}_{11}^{\boldsymbol{\left( 1 \right)}}\boldsymbol{B}_{11}^{\boldsymbol{\left( 2 \right)}}&0&0&0\\
0&0&0&0&0&0&0&\boldsymbol{b}_1^{\boldsymbol{\left( 1 \right)}}\boldsymbol{B}_{11}^{\boldsymbol{\left( 2 \right)}}&0&0&0&\boldsymbol{B}_{11}^{\boldsymbol{\left( 1 \right)}}\boldsymbol{B}_{11}^{\boldsymbol{\left( 2 \right)}}&0&0\\
0&0&0&0&0&0&0&\boldsymbol{B}_{11}^{\boldsymbol{\left( 1 \right)}}&0&0&0&0&0&0\\
0&0&0&0&0&0&0&\boldsymbol{B}_{11}^{\boldsymbol{\left( 2 \right)}}&0&0&0&0&0&0
\end{array}} \right]
\end{equation}
\begin{equation}
\label{Cs}
\boldsymbol{c_s} = \left[ {\begin{array}{*{20}{c}}
0&0&0&0&0&0&0&\left(\boldsymbol{b}_1^{\boldsymbol{\left( 1 \right)}} + \boldsymbol{b}_1^{\boldsymbol{\left( 2 \right)}} - \boldsymbol{b}_1^{\boldsymbol{\left( 1 \right)}}\boldsymbol{b}_1^{\boldsymbol{\left( 2 \right)}}\right)&0&0&\boldsymbol{b}_1^{\boldsymbol{\left( 1 \right)}}&\boldsymbol{b}_1^{\boldsymbol{\left( 2 \right)}}&\boldsymbol{b}_1^{\boldsymbol{\left( 1 \right)}}&\boldsymbol{b}_1^{\boldsymbol{\left( 2 \right)}}
\end{array}} \right]^\top
\end{equation}
\begin{equation}
\label{Cu}
\boldsymbol{c_u} = \left[ {\begin{array}{*{20}{c}}
0&\boldsymbol{b}_1^{\boldsymbol{\left( 2 \right)}}&\boldsymbol{b}_1^{\boldsymbol{\left( 1 \right)}}&0&0&0&0&0&0&0&0&0&0&0
\end{array}} \right]^\top
\end{equation}
\setcounter{equation}{\value{MYtempeqncnt}}
\hrulefill
\vspace*{4pt}
\end{figure*}

In addition to the transition probability matrix, we need to discuss the derivation of the initial probability vector, denoted as $\boldsymbol{\sigma}$, to complete the construction of $\left\{Y_h\right\}$.
The size of $\boldsymbol{\sigma}$ is also up to $M$, which makes it difficult to obtain its expression directly by observation.
Therefore, we construct a discrete-time recurrent Markov chain (RMC) $\left\{W_h\right\}$ to model the DTDQ system purely from the perspective of the cyclic alternation of servers between idle and busy, in contrast to $\left\{Y_h\right\}$ which is carefully designed to track the exact distribution of $\Delta$.
The state space of $\left\{W_h\right\}$ is shown in Table \ref{Table3}, where no transient or absorbing states are contained.
Moreover, Table \ref{Table4} presents the state transition probabilities of $\left\{W_h\right\}$.
By calculating the steady-state probabilities of $\left\{W_h\right\}$, we can further know the probability distribution of the system condition after $P^*$ is generated, i.e., the initial probability vector $\boldsymbol{\sigma}$ of $\left\{Y_h\right\}$.

According to Table \ref{Table1}, it can be found that the initial state of $\left\{Y_h\right\}$ can only be $\left(1,i,j,0\right)$, $\left(3,i,j,0\right)$, $\left(5,i,0,0\right)$, and $\left(6,0,j,0\right)$.
The reason is that when the system is initialized, i.e., when $P^*$ is generated, if there is already a packet in the system, the moment of its generation must be earlier than that of $P^*$.
Moreover, using the policy of $S_1$ priority or $S_2$ priority given in Table \ref{Table2}, there is only one of $\left(5,i,0,0\right)$, and $\left(6,0,j,0\right)$ that will be an initial state of $\left\{Y_h\right\}$.
To connect $\left\{Y_h\right\}$ and $\left\{W_h\right\}$, we next distinguish the transition events of $\left\{W_h\right\}$ that can be equivalent to the initialization of $\left\{Y_h\right\}$.
For the case of $\left(1,i,j,0\right)$, the initialization of $\left\{Y_h\right\}$ can be triggered by the transition from $\left(3,0,j,k-1\right)_\text R$ to $\left(4,i^\prime,j^\prime,0\right)_\text R$ with a probability of $\boldsymbol{B}^{\boldsymbol{\left(2\right)}}_{jj^\prime}\boldsymbol{\alpha} ^{\boldsymbol{\left(1\right)}}_{i^\prime}$ or the transition from $\left(4,i,j,k-1\right)_\text R$ to $\left(4,i^\prime,j^\prime,0\right)_\text R$ with a probability of $\boldsymbol{b}^{\boldsymbol{\left(1\right)}}_{i} \boldsymbol{B} ^{\boldsymbol{\left(2\right)}}_{jj^\prime} \boldsymbol{\alpha} ^{\boldsymbol{\left(1\right)}}_{i^\prime}$.
For the case of $\left(4,i,j,0\right)$, the initialization of $\left\{Y_h\right\}$ can be triggered by the transition from $\left(2,i,0,k-1\right)_\text R$ to $\left(4,i^\prime,j^\prime,0\right)_\text R$ with a probability of $\boldsymbol{B}^{\boldsymbol{\left(1\right)}}_{ii^\prime}\boldsymbol{\alpha} ^{\boldsymbol{\left(2\right)}}_{j^\prime}$ or the transition from $\left(4,i,j,k-1\right)_\text R$ to $\left(4,i^\prime,j^\prime,0\right)_\text R$ with a probability of $\boldsymbol{b}^{\boldsymbol{\left(2\right)}}_{j} \boldsymbol{B} ^{\boldsymbol{\left(1\right)}}_{ii^\prime} \boldsymbol{\alpha} ^{\boldsymbol{\left(2\right)}}_{j^\prime}$.
Note that although there are two events of transition from $\left(4,i,j,k-1\right)_\text R$ to $\left(4,i^\prime,j^\prime,0\right)_\text R$, they have different physical significance, where the one with probability $\boldsymbol{b}^{\boldsymbol{\left(1\right)}}_{i} \boldsymbol{B} ^{\boldsymbol{\left(2\right)}}_{jj^\prime} \boldsymbol{\alpha} ^{\boldsymbol{\left(1\right)}}_{i^\prime}$ and the one with probability $\boldsymbol{b}^{\boldsymbol{\left(2\right)}}_{j} \boldsymbol{B} ^{\boldsymbol{\left(1\right)}}_{ii^\prime} \boldsymbol{\alpha} ^{\boldsymbol{\left(2\right)}}_{j^\prime}$ corresponds to that the newly generated packet is served by $S_1$ and $S_2$, respectively.
We also make a distinction for this in Table \ref{Table4}.
Furthermore, the case of $\left(5,i,0,0\right)$ occurs for the policy of $S_1$ priority, where the initialization of $\left\{Y_h\right\}$ can be triggered by the transition from $\left(1,0,0,k-1\right)_\text R$ to $\left(2,i^\prime,0,0\right)_\text R$ with a probability of $\boldsymbol{\alpha} ^{\boldsymbol{\left(1\right)}}_{i^\prime}$, the transition from $\left(2,i,0,k-1\right)_\text R$ to $\left(2,i^\prime,0,0\right)_\text R$ with a probability of $\boldsymbol{b} ^{\boldsymbol{\left(1\right)}}_{i}\boldsymbol{\alpha} ^{\boldsymbol{\left(1\right)}}_{i^\prime}$, the transition from $\left(3,0,j,k-1\right)_\text R$ to $\left(2,i^\prime,0,0\right)_\text R$ with a probability of $\boldsymbol{b}^{\boldsymbol{\left(2\right)}}_{j}\boldsymbol{\alpha} ^{\boldsymbol{\left(1\right)}}_{i^\prime}$, or the transition from $\left(4,i,j,k-1\right)_\text R$ to $\left(2,i^\prime,0,0\right)_\text R$ with a probability of $\boldsymbol{b}^{\boldsymbol{\left(1\right)}}_{i} \boldsymbol{b} ^{\boldsymbol{\left(2\right)}}_{j} \boldsymbol{\alpha} ^{\boldsymbol{\left(1\right)}}_{i^\prime}$.
Symmetrically, the case of $\left(6,0,j,0\right)$ occurs for the policy of $S_2$ priority, where the initialization of $\left\{Y_h\right\}$ can be triggered by the transition from $\left(1,0,0,k-1\right)_\text R$ to $\left(3,0,j^\prime,0\right)_\text R$ with a probability of $\boldsymbol{\alpha} ^{\boldsymbol{\left(2\right)}}_{j^\prime}$, the transition from $\left(2,i,0,k-1\right)_\text R$ to $\left(3,0,j^\prime,0\right)_\text R$ with a probability of $\boldsymbol{b} ^{\boldsymbol{\left(1\right)}}_{i}\boldsymbol{\alpha} ^{\boldsymbol{\left(2\right)}}_{j^\prime}$, the transition from $\left(3,0,j,k-1\right)_\text R$ to $\left(3,0,j^\prime,0\right)_\text R$ with a probability of $\boldsymbol{b}^{\boldsymbol{\left(2\right)}}_{j}\boldsymbol{\alpha} ^{\boldsymbol{\left(2\right)}}_{j^\prime}$, or the transition from $\left(4,i,j,k-1\right)_\text R$ to $\left(3,0,j^\prime,0\right)_\text R$ with a probability of $\boldsymbol{b}^{\boldsymbol{\left(1\right)}}_{i} \boldsymbol{b} ^{\boldsymbol{\left(2\right)}}_{j} \boldsymbol{\alpha} ^{\boldsymbol{\left(2\right)}}_{j^\prime}$.

\begin{table}
\centering
\caption{State space of the discrete-time RMC $\left\{W_h\right\}$}
\renewcommand{\arraystretch}{1.5}
\label{Table3}
\begin{tabular}{|c|c|}
\hline
State & Description \\
\hline\hline
$\left(1,0,0,l\right)_\text R$ & $S_1$, $S_2$ idle \\
\hline
$\left(2,i,0,l\right)_\text R$ & $P_1$ on $S_1$, $S_2$ idle \\
\hline
$\left(3,0,j,l\right)_\text R$ & $P_2$ on $S_2$, $S_1$ idle \\
\hline
$\left(4,i,j,l\right)_\text R$ & $P_1$ on $S_1$, $P_2$ on $S_2$ \\
\hline
\end{tabular}
\end{table}

On the basis of the above, we can derive the probability of each initial state of $\left\{Y_h\right\}$ and thus $\boldsymbol\sigma$.
Denote the steady-state probabilities of the four classes of states for $\left\{W_h\right\}$ as $p\left(1,0,0,l\right)_\text R$, $p\left(2,i,0,l\right)_\text R$, $p\left(3,0,j,l\right)_\text R$, and $p\left(4,i,j,l\right)_\text R$, respectively.
Then, the probability that $\left\{Y_h\right\}$ starts in $\left(1,i^\prime,j^\prime,0\right)$ is given by
\begin{align}
\label{q1}
q\left(1,i^\prime,j^\prime,0\right) = &\sum\limits_{j = 1}^{{N_2}} {\boldsymbol{B}^{\boldsymbol{\left( 2 \right)}}_{jj^\prime}\boldsymbol{\alpha} ^{\boldsymbol{\left( 1 \right)}}_{i^\prime}p\left( {3,0,j,k-1} \right)_\text R}  \notag \\
&+ \sum\limits_{i = 1}^{{N_1}} {\sum\limits_{j = 1}^{{N_2}} {\boldsymbol{b}^{\boldsymbol{\left( 1 \right)}}_i\boldsymbol{B}^{\boldsymbol{\left( 2 \right)}}_{jj^\prime}\boldsymbol{\alpha} ^{\boldsymbol{\left( 1 \right)}}_{i^\prime}p\left( {4,i,j,k-1} \right)_\text R} },
\tag{8}
\end{align}
where $1\le i^\prime \le N_1$ and $1\le j^\prime \le N_2$.
Similarly, the probability that $\left\{Y_h\right\}$ starts in $\left(4,i^\prime,j^\prime,0\right)$ is given by
\begin{align}
\label{q4}
q\left(4,i^\prime,j^\prime,0\right) = &\sum\limits_{i = 1}^{{N_1}} {\boldsymbol{B}^{\boldsymbol{\left(1\right)}}_{ii^\prime}\boldsymbol{\alpha} ^{\boldsymbol{\left(2\right)}}_{j^\prime}p\left( {2,i,0,k-1} \right)_\text R} \notag \\
&+ \sum\limits_{i = 1}^{{N_1}} {\sum\limits_{j = 1}^{{N_2}} {\boldsymbol{b}^{\boldsymbol{\left(2\right)}}_{j} \boldsymbol{B} ^{\boldsymbol{\left(1\right)}}_{ii^\prime} \boldsymbol{\alpha} ^{\boldsymbol{\left(2\right)}}_{j^\prime}p\left( {4,i,j,k-1} \right)_\text R} }.
\tag{9}
\end{align}
The probability that $\left\{Y_h\right\}$ starts in $\left(5,i^\prime,0,0\right)$ is given by
\begin{align}
\label{q5}
&q\left(5,i^\prime,0,0\right) \notag \\
={} &\boldsymbol{\alpha} ^{\boldsymbol{\left(1\right)}}_{i^\prime}p\left( {1,0,0,k-1} \right)_\text R +\! \sum\limits_{i = 1}^{{N_1}} {\boldsymbol{b} ^{\boldsymbol{\left(1\right)}}_{i}\boldsymbol{\alpha} ^{\boldsymbol{\left(1\right)}}_{i^\prime}p\left( {2,i,0,k-1} \right)_\text R} \notag \\
&+ \sum\limits_{j = 1}^{{N_2}} {\boldsymbol{b} ^{\boldsymbol{\left(2\right)}}_{j}\boldsymbol{\alpha} ^{\boldsymbol{\left(1\right)}}_{i^\prime}p\left( {3,0,j,k-1} \right)_\text R} \notag \\
&+ \sum\limits_{i = 1}^{{N_1}} {\sum\limits_{j = 1}^{{N_2}} {\boldsymbol{b}^{\boldsymbol{\left(1\right)}}_{i} \boldsymbol{b} ^{\boldsymbol{\left(2\right)}}_{j} \boldsymbol{\alpha} ^{\boldsymbol{\left(1\right)}}_{i^\prime}p\left( {4,i,j,k-1} \right)_\text R} } .
\tag{10}
\end{align}
\begin{table}
\centering
\caption{Transition probabilities of the discrete-time RMC $\left\{W_h\right\}$}
\renewcommand{\arraystretch}{1.5}
\label{Table4}
\begin{tabular}{|p{1.284cm}<{\centering}|p{2.892cm}<{\centering}|p{1.063cm}<{\centering}|p{1.816cm}<{\centering}|}
\hline
From & \multicolumn{2}{c|}{To} & Probability \\
\hline\hline
\multirow{3}*{$\left(1,0,0,l\right)_\text R$} & \multicolumn{2}{c|}{$\left(1,0,0,l+1\right)_\text R$ for $l<k-1$} & $1$ \\
\cline{2-4}
~ & $S_1$ priority: $\left(2,i^\prime,0,0\right)_\text R$ & for & $\boldsymbol{\alpha} _{i^\prime}^{\boldsymbol{\left(1\right)}}$ \\
\cline{2-2}\cline{4-4}
~ & $S_2$ priority: $\left(3,0,j^\prime,0\right)_\text R$ & $l=k-1$ & $\boldsymbol{\alpha} _{j^\prime}^{\boldsymbol{\left(2\right)}}$ \\
\hline
\multirow{5}*{$\left(2,i,0,l\right)_\text R$} & \multicolumn{2}{c|}{$\left(1,0,0,l+1\right)_\text R$ for $l<k-1$} & $\boldsymbol{b} _{i}^{\boldsymbol{\left(1\right)}}$ \\
\cline{2-4}
~ & $S_1$ priority: $\left(2,i^\prime,0,0\right)_\text R$ & for & $\boldsymbol{b} _{i}^{\boldsymbol{\left(1\right)}}\boldsymbol{\alpha} _{i^\prime}^{\boldsymbol{\left(1\right)}}$ \\
\cline{2-2}\cline{4-4}
~ & $S_2$ priority: $\left(3,0,j^\prime,0\right)_\text R$ & $l=k-1$ & $\boldsymbol{b} _{i}^{\boldsymbol{\left(1\right)}}\boldsymbol{\alpha} _{j^\prime}^{\boldsymbol{\left(2\right)}}$ \\
\cline{2-4}
~ & \multicolumn{2}{c|}{$\left(2,i^\prime,0,l+1\right)_\text R$ for $l<k-1$} & $\boldsymbol{B} _{ii^\prime}^{\boldsymbol{\left(1\right)}}$ \\
\cline{2-4}
~ & \multicolumn{2}{c|}{$\left(4,i^\prime,j^\prime,0\right)_\text R$ for $l=k-1$} & $\boldsymbol{B} _{ii^\prime}^{\boldsymbol{\left(1\right)}}\boldsymbol{\alpha} _{j^\prime}^{\boldsymbol{\left(2\right)}}$ \\
\hline
\multirow{5}*{$\left(3,0,j,l\right)_\text R$} & \multicolumn{2}{c|}{$\left(1,0,0,l+1\right)_\text R$ for $l<k-1$} & $\boldsymbol{b} _{j}^{\boldsymbol{\left(2\right)}}$ \\
\cline{2-4}
~ & $S_1$ priority: $\left(2,i^\prime,0,0\right)_\text R$ & for & $\boldsymbol{b} _{j}^{\boldsymbol{\left(2\right)}}\boldsymbol{\alpha} _{i^\prime}^{\boldsymbol{\left(1\right)}}$ \\
\cline{2-2}\cline{4-4}
~ & $S_2$ priority: $\left(3,0,j^\prime,0\right)_\text R$ & $l=k-1$ & $\boldsymbol{b} _{j}^{\boldsymbol{\left(2\right)}}\boldsymbol{\alpha} _{j^\prime}^{\boldsymbol{\left(2\right)}}$ \\
\cline{2-4}
~ & \multicolumn{2}{c|}{$\left(3,0,j^\prime,l+1\right)_\text R$ for $l<k-1$} & $\boldsymbol{B} _{jj^\prime}^{\boldsymbol{\left(2\right)}}$ \\
\cline{2-4}
~ & \multicolumn{2}{c|}{$\left(4,i^\prime,j^\prime,0\right)_\text R$ for $l=k-1$} & $\boldsymbol{B} _{jj^\prime}^{\boldsymbol{\left(2\right)}}\boldsymbol{\alpha} _{i^\prime}^{\boldsymbol{\left(1\right)}}$ \\
\hline
\multirow{9}*{$\left(4,i,j,l\right)_\text R$} & \multicolumn{2}{c|}{$\left(3,0,j^\prime,l+1\right)_\text R$ for $l<k-1$} & $\boldsymbol{b} _{i}^{\boldsymbol{\left(1\right)}}\boldsymbol{B} _{jj^\prime}^{\boldsymbol{\left(2\right)}}\boldsymbol{\alpha} _{i^\prime}^{\boldsymbol{\left(1\right)}}$ \\
\cline{2-4}
~ & \multicolumn{2}{c|}{$\left(4,i^\prime,j^\prime,0\right)_\text R$ for $l=k-1$} & $\boldsymbol{b} _{i}^{\left(1\right)}\boldsymbol{B} _{jj^\prime}^{\boldsymbol{\left(2\right)}}\boldsymbol{\alpha} _{i^\prime}^{\boldsymbol{\left(1\right)}}$ \\
\cline{2-4}
~ & \multicolumn{2}{c|}{$\left(1,0,0,l+1\right)_\text R$ for $l<k-1$} & $\boldsymbol{b} _{i}^{\boldsymbol{\left(1\right)}}\boldsymbol{b} _{j}^{\boldsymbol{\left(2\right)}}$ \\
\cline{2-4}
~ & $S_1$ priority: $\left(2,i^\prime,0,0\right)_\text R$ & for & $\boldsymbol{b} _{i}^{\boldsymbol{\left(1\right)}}\boldsymbol{b} _{j}^{\boldsymbol{\left(2\right)}}\boldsymbol{\alpha} _{i^\prime}^{\boldsymbol{\left(1\right)}}$ \\
\cline{2-2}\cline{4-4}
~ & $S_2$ priority: $\left(3,0,j^\prime,0\right)_\text R$ & $l=k-1$ & $\boldsymbol{b} _{i}^{\boldsymbol{\left(1\right)}}\boldsymbol{b} _{j}^{\boldsymbol{\left(2\right)}}\boldsymbol{\alpha} _{j^\prime}^{\boldsymbol{\left(2\right)}}$ \\
\cline{2-4}
~ & \multicolumn{2}{c|}{$\left(2,i^\prime,0,l+1\right)_\text R$ for $l<k-1$} & $\boldsymbol{b} _{j}^{\boldsymbol{\left(2\right)}}\boldsymbol{B} _{ii^\prime}^{\boldsymbol{\left(1\right)}}$ \\
\cline{2-4}
~ & \multicolumn{2}{c|}{$\left(4,i^\prime,j^\prime,0\right)_\text R$ for $l=k-1$} & $\boldsymbol{b} _{j}^{\boldsymbol{\left(2\right)}}\boldsymbol{B} _{ii^\prime}^{\boldsymbol{\left(1\right)}}\boldsymbol{\alpha} _{j^\prime}^{\boldsymbol{\left(2\right)}}$ \\
\cline{2-4}
~ & \multicolumn{2}{c|}{$\left(4,i^\prime,j^\prime,l+1\right)_\text R$ for $l<k-1$} & \multirow{2}*{$\boldsymbol{B} _{ii^\prime}^{\boldsymbol{\left(1\right)}}\boldsymbol{B} _{jj^\prime}^{\boldsymbol{\left(2\right)}}$} \\
\cline{2-3}
~ & \multicolumn{2}{c|}{$\left(4,i^\prime,j^\prime,k-1\right)_\text R$ for $l=k-1$} & ~ \\
\hline
\end{tabular}
\end{table}
The probability that $\left\{Y_h\right\}$ starts in $\left(6,0,j^\prime,0\right)$ is given by
\begin{align}
\label{q6}
&q\left(6,0,j^\prime,0\right) \notag \\
= {} &\boldsymbol{\alpha} ^{\boldsymbol{\left(2\right)}}_{j^\prime}p\left( {1,0,0,k-1} \right)_\text R + \sum\limits_{i = 1}^{{N_1}} {\boldsymbol{b} ^{\boldsymbol{\left(1\right)}}_{i}\boldsymbol{\alpha} ^{\boldsymbol{\left(2\right)}}_{j^\prime}p\left( {2,i,0,k-1} \right)_\text R} \notag \\
&+ \sum\limits_{j = 1}^{{N_2}} {\boldsymbol{b} ^{\boldsymbol{\left(2\right)}}_{j}\boldsymbol{\alpha} ^{\boldsymbol{\left(2\right)}}_{j^\prime}p\left( {3,0,j,k-1} \right)_\text R} \notag \\
&+ \sum\limits_{i = 1}^{{N_1}} {\sum\limits_{j = 1}^{{N_2}} {\boldsymbol{b}^{\boldsymbol{\left(1\right)}}_{i} \boldsymbol{b} ^{\boldsymbol{\left(2\right)}}_{j} \boldsymbol{\alpha} ^{\boldsymbol{\left(2\right)}}_{j^\prime}p\left( {4,i,j,k-1} \right)_\text R} }.
\tag{11}
\end{align}
To obtain $\boldsymbol\sigma$, we need to further perform a normalization of (\ref{q1}), (\ref{q4}), (\ref{q5}), and (\ref{q6}).
For the example of $S_1$ priority, $\sigma \left( {1,i^\prime,j^\prime,0} \right)$ is equal to the following expression,
\begin{align}
\label{sigma1}
{}&\frac{q\left(1,i^\prime,j^\prime,0\right)}{\sum\nolimits_{i^\prime =1}^{{N_1}} {\sum\nolimits_{j^\prime=1}^{{N_2}} \left( {q\left(1,i^\prime,j^\prime,0\right) + q\left(4,i^\prime,j^\prime,0\right) + q\left(5,i^\prime,0,0\right)} \right)  } },
\tag{12}
\end{align}
which is the element of $\boldsymbol\sigma$ corresponding to $\left(1,i^\prime,j^\prime,0\right)$.
Through the same method, $\sigma \left( {4,i^\prime,j^\prime,0} \right)$ and $\sigma \left( {5,i^\prime,0,0} \right)$ can also be derived.
Except for these three classes of elements, all other elements of $\boldsymbol\sigma$ are zero.
Again using $M=14$ as an example, (\ref{sigma}) shows the exact expression of $\boldsymbol\sigma$, where the non-zero elements are determined by (\ref{q1M14}), (\ref{q4M14}), and (\ref{q5M14}).
For the policy of $S_2$ priority, $q\left( {5,i^\prime,0,0} \right)$ in (\ref{sigma1}) needs to be replaced by $q\left( {6,0,j^\prime,0} \right)$.
Also, the three classes of non-zero elements are $\sigma \left( {1,i^\prime,j^\prime,0} \right)$, $\sigma \left( {4,i^\prime,j^\prime,0} \right)$, and $\sigma \left( {6,0,j^\prime,0} \right)$.

\section{Exact Distributions of AoI and PAoI for DTDQ Status Update Systems}\label{AoIderivation}
Based on the constructed discrete-time AMC, we further derive the PMF of the steady-state AoI $\Delta$.
Without loss of generality, we still choose Cycle $n$ in Fig.~\ref{AoI_evolution} as the object of study.
Recall that in Fig.~\ref{AoI_evolution}, the red portion corresponds to non-$P^*$ transient states ($\left(7,0,0,l\right)$, $\left(8,i,j\right)$, $\left(9,i,0,l\right)$, $\left(10,0,j,l\right)$, $\left(11,i,j,l\right)$, $\left(12,i,j,l\right)$, $\left(13,i,0,l\right)$, and $\left(14,0,j,l\right)$) of $\left\{Y_h\right\}$.
Denote the subset composed of these non-$P^*$ transient states, i.e., Subset 2, as $\mathcal{C}$.
For Cycle $n$, $\left\{Y_h\right\}$ is initiated at the instant $g_{n-1}$.
The exact expression of the PMF ${u_\Delta }\left( h \right)$ is presented in Theorem~\ref{Theorem1}.

\begin{figure*}[!t]
\normalsize
\setcounter{MYtempeqncnt}{\value{equation}}
\setcounter{equation}{12}
\begin{equation}
\label{sigma}
\boldsymbol\sigma  = \left[ {\begin{array}{*{20}{c}}
{\sigma \left( {1,1,1,0} \right)}&0&0&{\sigma \left( {4,1,1,0} \right)}&{\sigma \left( {5,1,0,0} \right)}&0&0&0&0&0&0&0&0&0
\end{array}} \right]
\end{equation}
\begin{equation}
\label{q1M14}
q\left( {1,1,1,0} \right) = \boldsymbol{B}_{11}^{\boldsymbol{\left( 2 \right)}}\boldsymbol{\alpha} _1^{\boldsymbol{\left( 1 \right)}}p{\left( {3,0,1,0} \right)_\text R} + \boldsymbol{b}_1^{\boldsymbol{\left( 1 \right)}}\boldsymbol{B}_{11}^{\boldsymbol{\left( 2 \right)}}\boldsymbol{\alpha} _1^{\boldsymbol{\left( 1 \right)}}p{\left( {4,1,1,0} \right)_\text R}
\end{equation}
\begin{equation}
\label{q4M14}
q\left( {4,1,1,0} \right) = \boldsymbol{B}_{11}^{\boldsymbol{\left( 1 \right)}}\boldsymbol{\alpha} _1^{\boldsymbol{\left( 2 \right)}}p{\left( {2,1,0,0} \right)_\text R} + \boldsymbol{b}_1^{\boldsymbol{\left( 2 \right)}}\boldsymbol{B}_{11}^{\boldsymbol{\left( 1 \right)}}\boldsymbol{\alpha} _1^{\boldsymbol{\left( 2 \right)}}p{\left( {4,1,1,0} \right)_\text R}
\end{equation}
\begin{equation}
\label{q5M14}
q\left( {5,1,0,0} \right)=
\boldsymbol{\alpha} _1^{\boldsymbol{\left( 1 \right)}}p{\left( {1,0,0,0} \right)_\text R} + b_1^{\left( 1 \right)}\alpha _1^{\left( 1 \right)}p{\left( {2,1,0,0} \right)_\text R} + \boldsymbol{b}_1^{\boldsymbol{\left( 2 \right)}}\boldsymbol{\alpha} _1^{\boldsymbol{\left( 1 \right)}}p{\left( {3,0,1,0} \right)_\text R} + \boldsymbol{b}_1^{\boldsymbol{\left( 1 \right)}}\boldsymbol{b}_1^{\boldsymbol{\left( 2 \right)}}\boldsymbol{\alpha} _1^{\boldsymbol{\left( 1 \right)}}p{\left( {4,1,1,0} \right)_\text R}
\end{equation}
\setcounter{equation}{\value{MYtempeqncnt}}
\hrulefill
\vspace*{4pt}
\end{figure*}

\begin{thm}\label{Theorem1}
From the perspective of discrete-time AMC modeling, the PMF of the steady-state AoI at a given value $h$ is proportional to the probability that $\left\{Y_h\right\}$ is visiting a non-$P^*$ transient state after $h$ slots since its initialization, 
\begin{align}
\label{Thm1}
{u_\Delta }\left( h \right) \propto \Pr \left( Y_h \in \mathcal{C} \right).
\tag{17}
\end{align}
Incorporating the transition matrix and the initial vector of $\left\{Y_h\right\}$, the exact expression of ${u_\Delta }\left( h \right)$ is further given by
\begin{align}
\label{Thm2}
{u_\Delta }\left( h \right)  = \frac{\boldsymbol{\sigma}\boldsymbol{A}^h\boldsymbol{v}}{{\boldsymbol{\sigma} \boldsymbol{A}{{\left( {\boldsymbol{I} - \boldsymbol{A}} \right)}^{ - 1}}\boldsymbol{v}}}.
\tag{18}
\end{align}
where $\boldsymbol{v}$ is a $M \times 1$ vector whose elements corresponding to non-$P^*$ states are one and all other elements are zero.\footnote{Similar to (\ref{QforM14}) and (\ref{sigma}), taking the case of $M=14$ as an example, we have $\boldsymbol{v} = \left[ {\begin{array}{*{20}{c}}
0&0&0&0&0&0&1&1&1&1&1&1&1&1
\end{array}} \right]^\top$.}
\end{thm}
\begin{proof}
For a given cycle of AoI, such as Cycle $n$ in Fig.~\ref{AoI_evolution}, the start and end time points are the receptions of the $(n-1)$st and $n$th up-to-date packets, respectively.
On the other hand, the discrete-time AMC $\{Y_h\}$ is initiated at $h=0$ with the generation of the $n-1$st up-to-date packet, the operation of which continues until the reception of the $n$th up-to-date packet, i.e., the absorption into the successful absorbing state $15$.
After the reception of the $(n-1)$st up-to-date packet, the AMC $\{Y_h\}$ is either visiting the subset $\mathcal{C}$ composed of the non-$P^*$ states or is absorbed.
Then, using sample path arguments, there is an AMC cycle corresponding to each AoI cycle, while the part of AMC cycle spent in $\mathcal{C}$ overlaps the AoI cycle.
Furthermore, a state in $\mathcal{C}$ is to be visited in the AMC cycle at time $h$ if the AoI value $h$ is visited in the AoI cycle.
Consequently, the PMF ${u_\Delta }\left( h \right)$ is equal to the probability that $Y_h \in \mathcal{C}$ divided by a fixed scale factor $\kappa_1$, ensuring that $\sum\nolimits_{h=1}^\infty {{u_\Delta }\left( h \right)}  = 1$.
That is,
\begin{align}
\label{proof1}
{u_\Delta }\left( h \right) = \kappa_1 \Pr \left(Y_h \in \mathcal{C} \right),
\tag{19}
\end{align}
which yields (\ref{Thm1}).
Next, we further perform normalization on $\Pr \left(Y_h \in \mathcal{C} \right)$ to obtain $\kappa_1$.
According to the definition of the discrete-time AMC, the probability that $\left\{Y_h\right\}$ is visiting a state belonging to $\mathcal{C}$ after $h$ steps is given by
\begin{align}
\label{proof2}
\Pr \left( {{Y_h} \in \mathcal C} \right) = \boldsymbol{\sigma}\boldsymbol{A}^h\boldsymbol{v}.
\tag{20}
\end{align}
In discrete-time systems, the minimum of AoI is $1$. 
Then, the range $h=1,2,\ldots$ is used to normalize $\Pr \left(Y_h \in \mathcal{C} \right)$.
Based on (\ref{proof2}), we have
\begin{align}
\label{proof3}
\sum\limits_{h = 1}^\infty  {\Pr \left( {{Y_h} \in \mathcal C} \right)} = \boldsymbol{\sigma} \sum\limits_{h = 1}^\infty  {{\boldsymbol{A}^h}} \boldsymbol{v} = \boldsymbol{\sigma} \boldsymbol{A}{{\left( {\boldsymbol{I} - \boldsymbol{A}} \right)}^{ - 1}}\boldsymbol{v}.
\tag{21}
\end{align}
Moreover, combining (\ref{proof1}), (\ref{proof3}), and $\sum\nolimits_h {{u_\Delta }\left( h \right)}  = 1$, it can be obtained that
\begin{align}
\label{proof4}
\kappa_1  = \frac{1}{\sum\nolimits_{h = 1}^\infty  {\Pr \left( {{Y_h} \in \mathcal C} \right)}} = \frac{1}{\boldsymbol{\sigma} \boldsymbol{A}{{\left( {\boldsymbol{I} - \boldsymbol{A}} \right)}^{ - 1}}\boldsymbol{v}}.
\tag{22}
\end{align}
By substituting (\ref{proof2}) and (\ref{proof4}) into (\ref{proof1}), (\ref{Thm2}) follows.
This completes the proof.
\end{proof}

Note that in order to model the distribution of steady-state AoI, in Theorem \ref{Theorem1} we focus only on the case where $\left\{Y_h\right\}$ is absorbed into the state $15$.
This is sufficient since the absorption of $\left\{Y_h\right\}$ into the state $16$ will not be shown in Fig.~\ref{AoI_evolution}.
Using the PMF given by (\ref{Thm2}), the average AoI of the investigated DTDQ system can be presented as
\begin{align}
\label{AoIbar}
\mathbb E\left[ \Delta\right] = \sum\limits_{h = 1}^\infty  {h{u_\Delta }\left( h \right)} = \frac{\boldsymbol{\sigma}\boldsymbol{A}{{\left( {\boldsymbol{I} - \boldsymbol{A}} \right)}^{ - 2}}\boldsymbol{v}}{{\boldsymbol{\sigma} \boldsymbol{A}{{\left( {\boldsymbol{I} - \boldsymbol{A}} \right)}^{ - 1}}\boldsymbol{v}}}.
\tag{23}
\end{align}
In addition, the average squared AoI is given by
\begin{align}
\label{AoI2bar}
\mathbb E\left[ \Delta^2\right] = \sum\limits_{h = 1}^\infty  {h^2{u_\Delta }\left( h \right)} = \frac{\boldsymbol{\sigma}\boldsymbol{A}{{\left( {\boldsymbol{I} - \boldsymbol{A}} \right)}^{ - 3}}\left({\boldsymbol{A}} + \boldsymbol{I} \right)\boldsymbol{v}}{{\boldsymbol{\sigma} \boldsymbol{A}{{\left( {\boldsymbol{I} - \boldsymbol{A}} \right)}^{ - 1}}\boldsymbol{v}}}.
\tag{24}
\end{align}
Besides $\mathbb E\left[ \Delta\right]$ and $\mathbb E\left[ \Delta^2\right]$, the exact expression for any other order moment of $\Delta$ can also be obtained by using the PMF derived in Theorem \ref{Theorem1}.

Similarly, we derive the PMF of PAoI from the perspective of discrete-time AMC modeling.
It can be seen in Fig.~\ref{AoI_evolution} that for Cycle $n$,
$\Phi_n$ is equal to the total time that $\left\{Y_h\right\}$ has experienced minus $1$ when it is absorbed into the state $15$.
Then, the PMF of the steady-state PAoI at a given value $h$ is denoted as
\begin{align}
\label{Phi}
{u_\Phi }\left( h \right) = {}& U_\Phi\left(h\right)-U_\Phi\left(h-1\right) \notag \\ 
={}&\kappa_2\left(\Pr \left( {{Y_{h+1}} = 15} \right)-\Pr \left( {{Y_h} = 15} \right)\right) \notag \\
= {}& \kappa_2\boldsymbol {\sigma}\boldsymbol{A}^{h}\boldsymbol{c_s},
\tag{25}
\end{align}
where $U_\Phi\left(h\right)=\kappa_2\Pr \left( {{Y_{h+1}} = 15} \right)=\kappa_2\boldsymbol {\sigma}{\sum\nolimits_{\xi = 1}^h  {\boldsymbol{A}^{\xi}}}\boldsymbol{c_s}$ is the cumulative distribution function of $\Phi$ and $\kappa_2$ is a fixed scale factor to ensure that $\lim_{h \rightarrow \infty} U_\Phi\left(h\right)=1$, i.e., $\sum\nolimits_{h=1}^\infty {{u_\Phi }\left( h \right)}  = 1$.
Based on (\ref{Phi}), we have
\begin{align}
\label{Kappa1}
\kappa_2  = \frac{1}{{\boldsymbol{\sigma} \sum\nolimits_{h = 1}^\infty  {{\boldsymbol{A}^h}} \boldsymbol{c_s}}} =
\frac{1}{\boldsymbol{\sigma} \boldsymbol{A}{{\left( {\boldsymbol{I} - \boldsymbol{A}} \right)}^{ - 1}}\boldsymbol{c_s}}.
\tag{26}
\end{align}
By substituting (\ref{Kappa1}) into (\ref{Phi}), the exact expression of ${u_\Phi }\left( h \right)$ is obtained as
\begin{align}
\label{Phi1}
{u_\Phi }\left( h \right)  = \frac{\boldsymbol{\sigma}\boldsymbol{A}^h\boldsymbol{c_s}}{{\boldsymbol{\sigma} \boldsymbol{A}{{\left( {\boldsymbol{I} - \boldsymbol{A}} \right)}^{ - 1}}\boldsymbol{c_s}}}.
\tag{27}
\end{align}
We can obtain an arbitrary order moment of $\Phi$ on the basis of (\ref{Phi1}).
In particular, the average PAoI is given by
\begin{align}
\label{PAoIbar}
\mathbb E\left[ \Phi\right] = \sum\limits_{h = 1}^\infty  {h{u_\Phi }\left( h \right)} = \frac{\boldsymbol{\sigma}\boldsymbol{A}{{\left( {\boldsymbol{I} - \boldsymbol{A}} \right)}^{ - 2}}\boldsymbol{c_s}}{{\boldsymbol{\sigma} \boldsymbol{A}{{\left( {\boldsymbol{I} - \boldsymbol{A}} \right)}^{ - 1}}\boldsymbol{c_s}}},
\tag{28}
\end{align}
which, as an important complement to $\mathbb E\left[ \Delta\right]$, can help to better understand the timeliness performance of the DTDQ system.
\section{Numerical Results}\label{Results}
In this section, we first provide numerical examples to demonstrate the validation of the constructed discrete-time AMC model.
Subsequently, numerical examples revealing the optimum freezing parameter that minimizes the average AoI are presented using geometrically, uniformly, and triangularly distributed service times.
The DPH-type representations of uniform and triangular distributions are obtained based on (\ref{DPHrepresentation}).
In addition, we also present several examples showing the reduction in average AoI due to freezing in cases where the service time distributions of the two servers are non-identical.

\subsection{Validation of the AMC Model}\label{Validation}
In Figs.~\ref{AverageAoI_ET}-\ref{AveragePAoI_ET}, we present the results obtained with the analytical (or theoretical) model and simulations, for the average AoI, the average squared AoI, and the average PAoI (as a function of the mean service time) for which the theoretical results are obtained from (\ref{AoIbar}), (\ref{AoI2bar}), and (\ref{PAoIbar}), respectively, given in Section \ref{AoIderivation}.
The simulation results are the output of Monte Carlo simulations of the DTDQ system for $4 \times 10^6$ slots.
Identically distributed service times are considered and two well-known DPH-type distributions, geometric and uniform, are used to model the service time of $S_m, m=1,2$.

\begin{figure}
\centering
\includegraphics[width=0.47\textwidth]{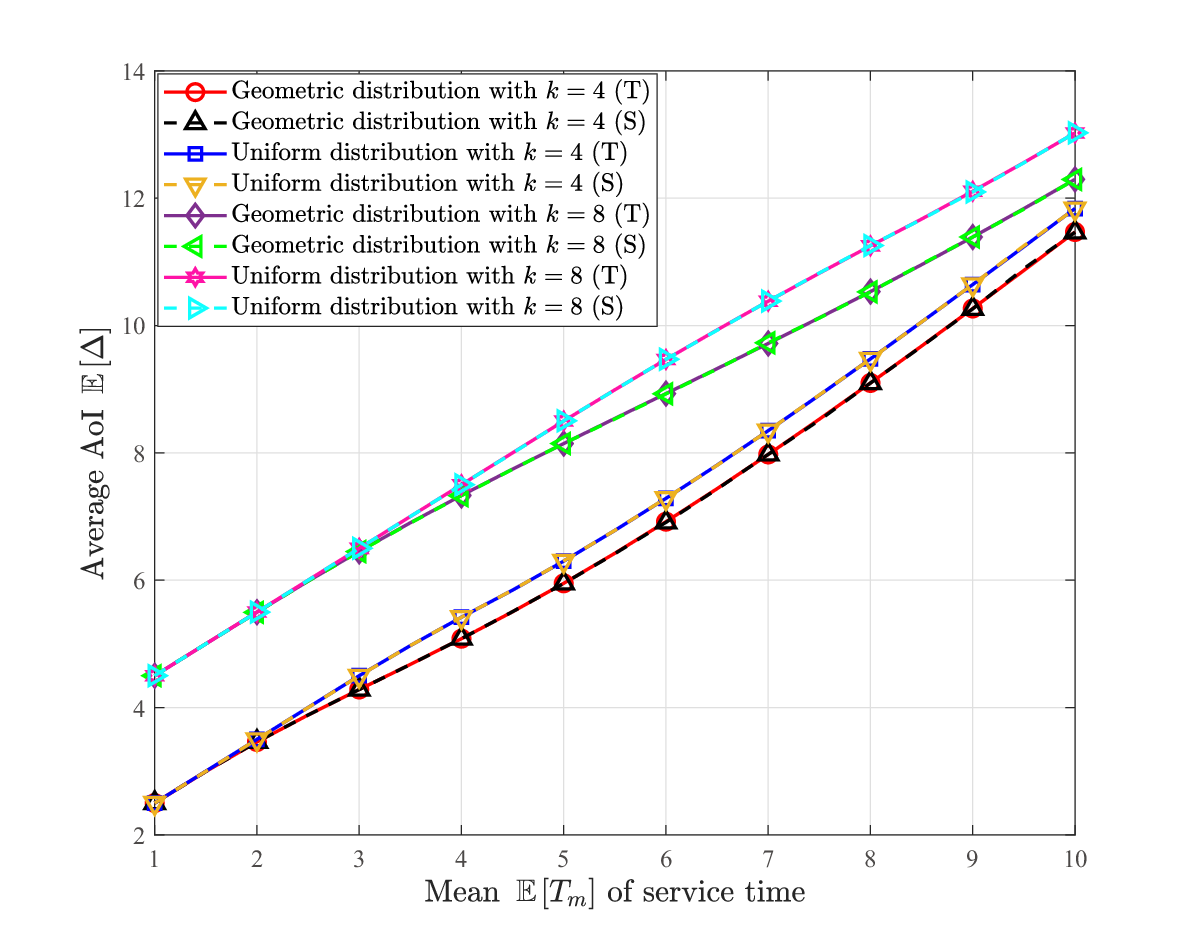}
\caption{Average AoI with respect to (wrt) the mean service time for geometric and uniform distributions. The service times $T_1$ and $T_2$ are identically distributed. (T) and (S) denote  the theoretical and simulation results, respectively.}
\label{AverageAoI_ET}
\end{figure}

\begin{figure}
\centering
\includegraphics[width=0.47\textwidth]{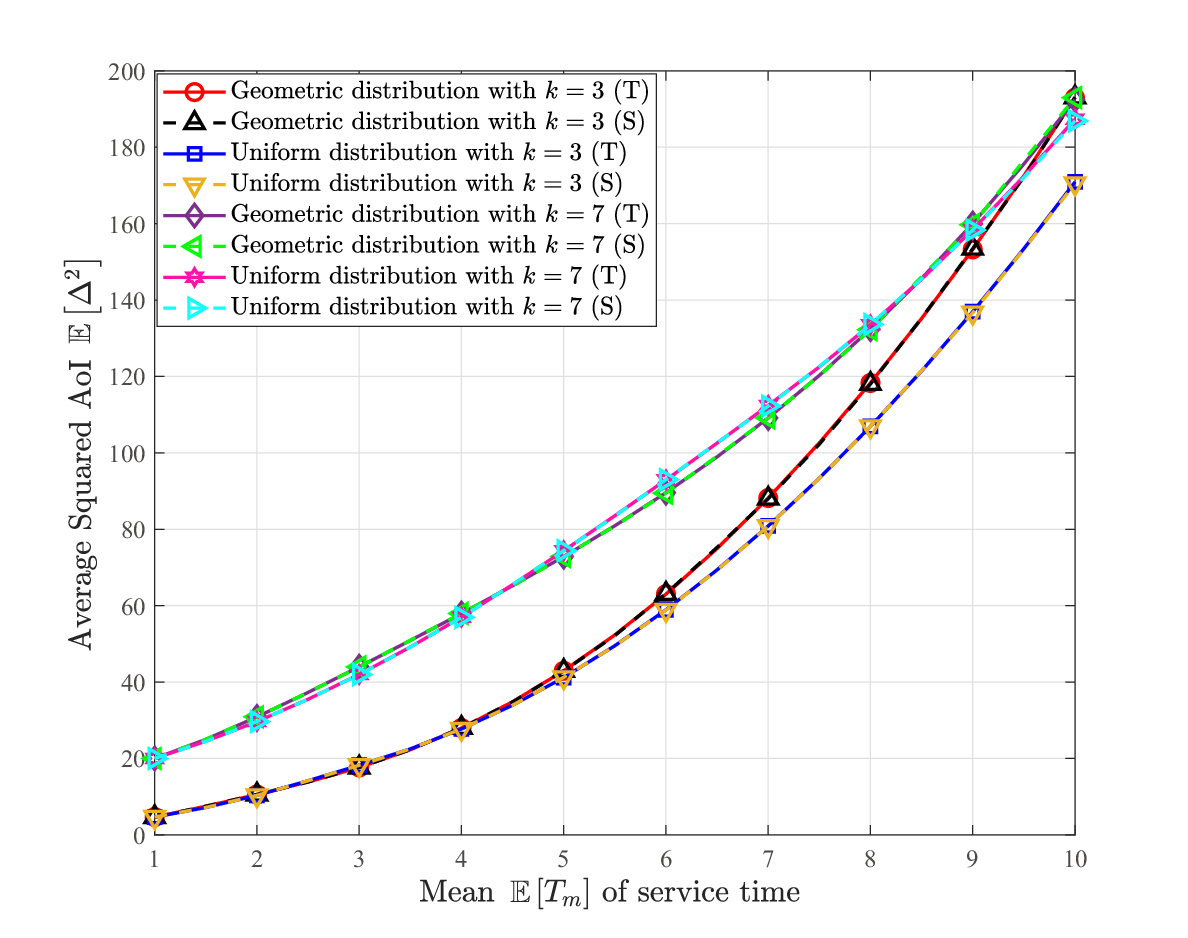}
\caption{Average aquared AoI wrt the mean service time for geometric and uniform distributions. The service times $T_1$ and $T_2$ are identically distributed. (T) and (S) denote the theoretical and simulation results, respectively.}
\label{AverageSquaredAoI_ET}
\end{figure}

In can be seen from Figs. \ref{AverageAoI_ET}-\ref{AveragePAoI_ET} that for different timeliness metrics and varying freezing parameters, the theoretical results always matched the simulation results very well, validating the correctness of Theorem~\ref{Theorem1} and the effectiveness of the constructed discrete-time AMC model.
In addition, we observe that both the average AoI and the average PAoI increase with an increase in the mean service time.
This is intuitively easy to understand, as longer service times mean that status update packets are less fresh when they reach the monitor, leading to worsened information timeliness.
We also observe that for $\mathbb E \left[T_m\right] = 1$, the curves of the geometric distribution and uniform distribution coincide exactly.
The reason is that at such a point, both distributions degenerate into the deterministic distribution located at one.
Focusing on the average AoI, we find that the performance gap between the geometric and uniform distributions is affected by the freezing parameter.
In Fig.~\ref{AverageAoI_ET}, for both $k=4$ and $k=8$, the average AoI corresponding to geometrically distributed service times is smaller than that corresponding to uniformly distributed service times.
However, it can be deduced from Fig.~\ref{AverageSquaredAoI_ET} that for $k = 3$, the average AoI corresponding to the geometric distribution is larger, while for $k = 7$, the average AoIs corresponding to the two distributions overlap several times.
Although Fig.~\ref{AverageAoI_ET} or Fig.~\ref{AverageSquaredAoI_ET} in isolation might suggest that a larger freezing parameter implies a larger average AoI, this is not actually always the case.
A counterexample is that for geometrically distributed service times with $\mathbb E \left[T_m\right] = 10$, the average AoI given in Fig.~\ref{AverageAoI_ET} is less than $13$ when $k=8$, but the average squared AoI shown in Fig.~\ref{AverageSquaredAoI_ET} is significantly larger than $169$ when $k=7$.
We leave the detailed discussion of the optimum freezing parameter for the next subsection.
Having demonstrated the effectiveness of the constructed discrete-time AMC model, from now on, we will resort to the theoretical results only, in the remaining numerical examples.

\begin{figure}
\centering
\includegraphics[width=0.47\textwidth]{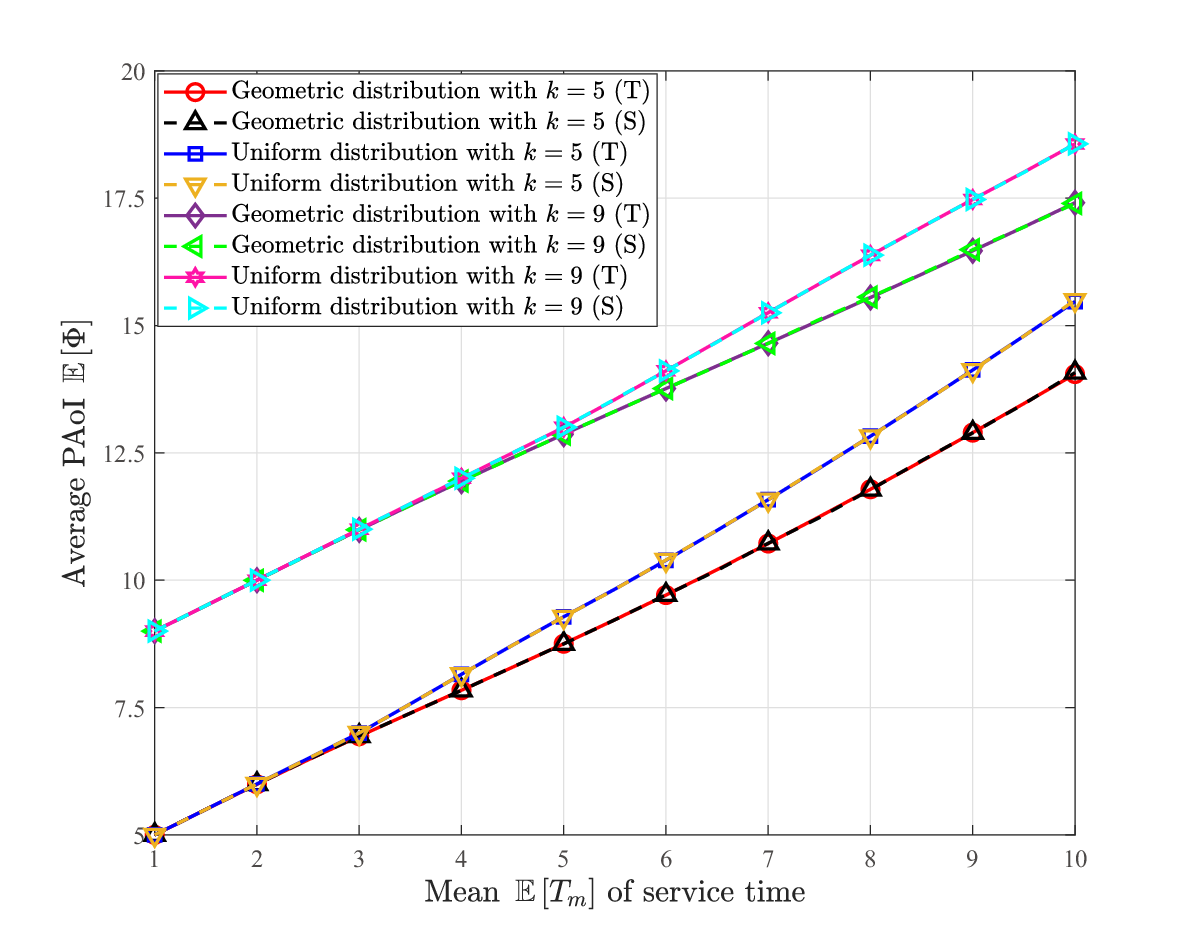}
\caption{Average PAoI versus mean service time for geometric and uniform distributions. The service times $T_1$ and $T_2$ are identically distributed. (T) and (S) denote theoretical and simulation results, respectively.}
\label{AveragePAoI_ET}
\end{figure}
\subsection{Optimum Freezing Parameter wrt Average AoI}\label{OptimumFreezing}
\begin{figure*}
\centering
\subfloat[]{
\includegraphics[width = 0.315\textwidth]{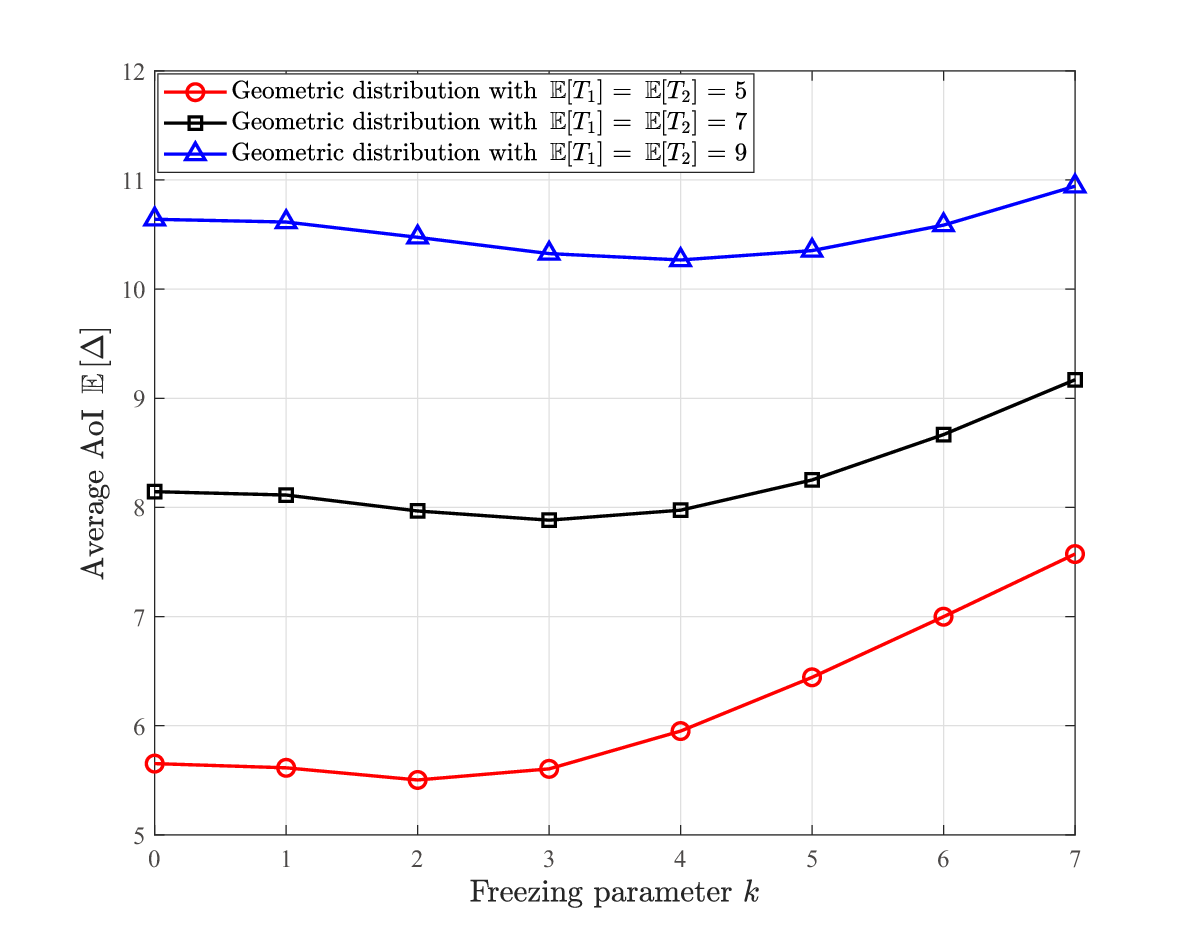}}
\subfloat[]{
\includegraphics[width = 0.322\textwidth]{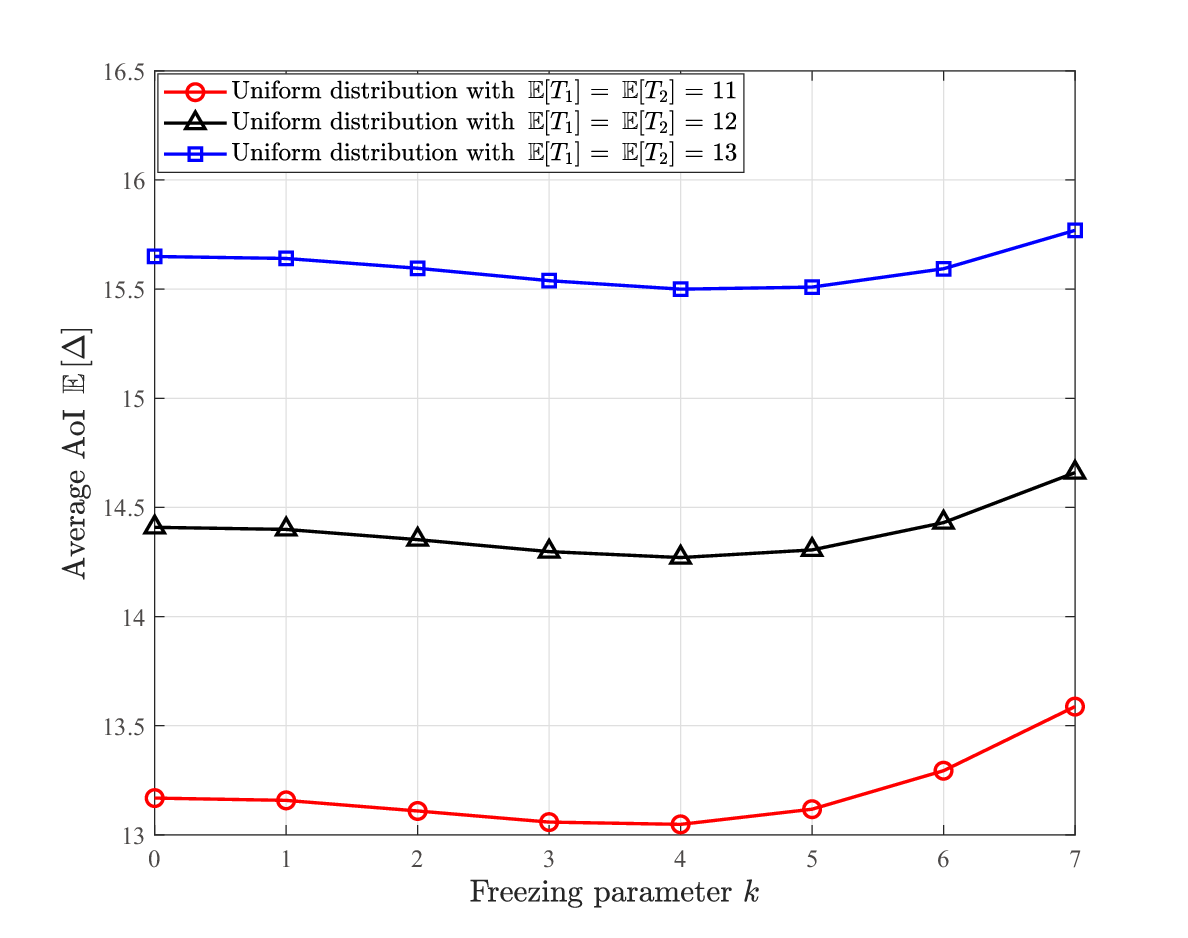}}
\subfloat[]{
\includegraphics[width = 0.337\textwidth]{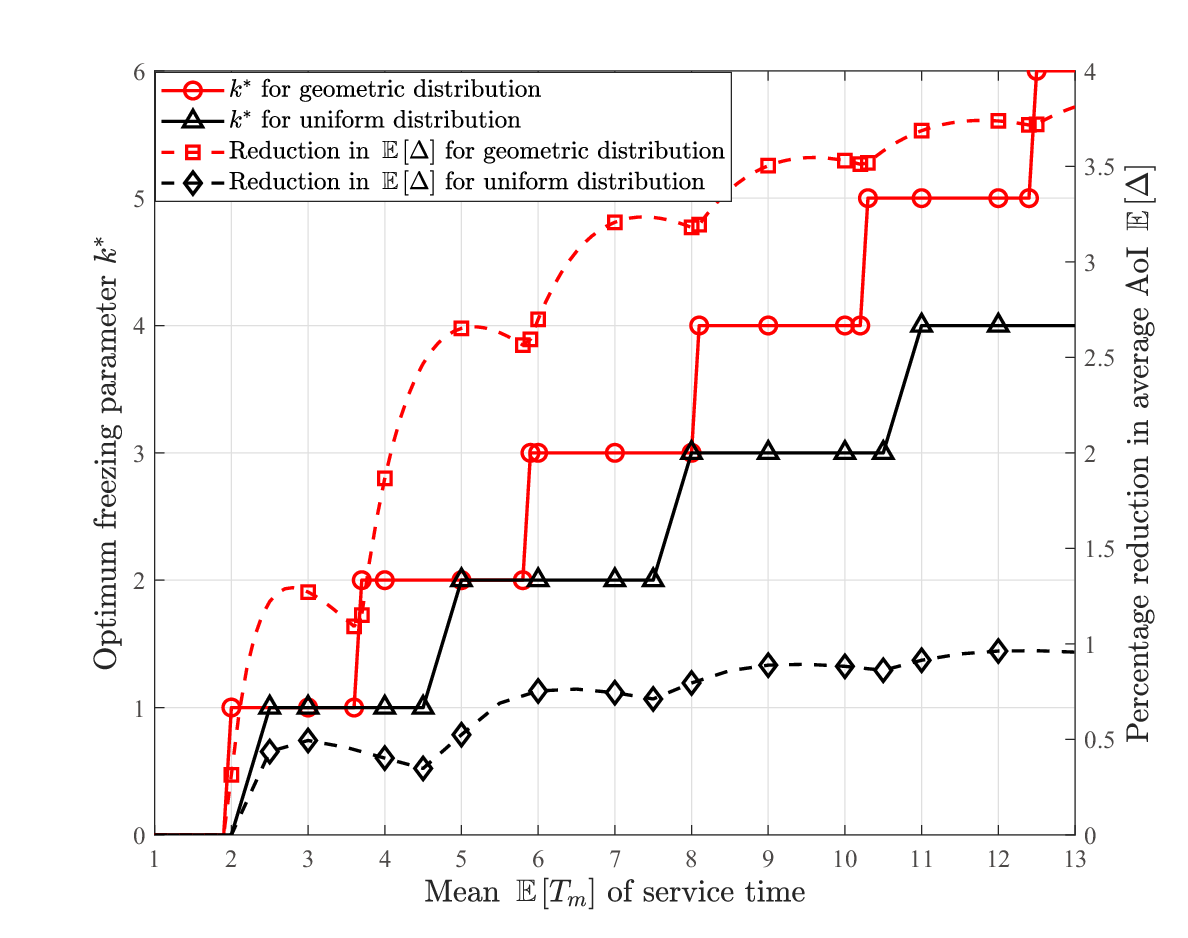}}
\caption{(a) Average AoI versus freezing parameter for geometric distribution. (b) Average AoI versus freezing parameter for uniform distribution. (c) Optimum freezing parameter and corresponding percentage reduction in average AoI versus mean service time. The service times $T_1$ and $T_2$ are identically distributed.}
\label{AverageAoI_k}
\end{figure*}

It has been recognized in Section \ref{OptimumFreezing} that the freezing parameter has an effect on the average AoI.
Next, we show the existence of the optimum freezing parameter and further discuss the relationship between the optimum freezing parameter and the statistical characteristics of service times.

Fig.~\ref{AverageAoI_k} presents the numerical examples using geometrically and uniformly distributed service times, for which we focus on the role of their mean.
From Fig.~\ref{AverageAoI_k}(a) and Fig.~\ref{AverageAoI_k}(b), it is observed that the curve of average AoI first drops and then rises with an increase of the freezing parameter $k$.
When the freezing parameter is small, both servers are likely to always remain busy.
Status update packets with proximate samples occupy the system's transmission resources, and the servers are unable to sample the source processes in a manner most favorable to timeliness.
This kind of timeliness penalty stems from proximate sampling and transmission.
For a given small freezing parameter, it is clear that the longer the mean service time, the greater the penalty of timeliness.
When the freezing parameter is quite large, it is likely that the source process will not be sampled in a timely manner even when a server is available. The timeliness of the system is also penalized at this point since the monitor would then receives status update packets at a lower frequency.
In Fig.~\ref{AverageAoI_k}(c), we depict the optimum freezing parameter $k^*$ versus the mean service time.
Using the optimum freezing parameter, the percentage reduction in average AoI compared to the no-freezing case (i.e., $k=0$) is also exhibited.
Note that in our model with freezing, if there is no free server once the freezing clock exceeds $k-1$ slots, sampling is not performed until a server becomes available.
On the other hand, we observe from Fig.~\ref{AverageAoI_k}(c) that the introduction of freezing does not always reduce the average AoI.
Specifically, the optimum freezing parameter is zero until $\mathbb E \left[T_m\right] \ge 2$ and $\mathbb E \left[T_m\right] \ge 2.5$ for geometrically and uniformly distributed service times, respectively.
This is reasonable because when the average service time is short enough, more frequent sampling and transmission can lead to a smaller average AoI, without considering the transmission cost.
As the average service time increases, the optimum freezing parameter increases.
One interesting finding is that for both distributions, the interval length of the average service time corresponding to each value of the optimal freezing parameter is basically fixed.
The interval is about $2.1$ for the geometric distribution and $2.5$ for the uniform distribution.
We also observe that within the interval, the gain from freezing first increases and then decreases slightly, which is due to the discrete-time nature of the system itself.
In addition, since the mean of a discrete uniform distribution can only be a multiple of $0.5$, the number of its available values within the same interval of mean is sparser than that for a geometric distribution.
Therefore, the rising edge of the $k^*$ curve corresponding to the uniform distribution is not very steep in Fig.~\ref{AverageAoI_k}(c).

\begin{figure}
\centering
\includegraphics[width=0.48\textwidth]{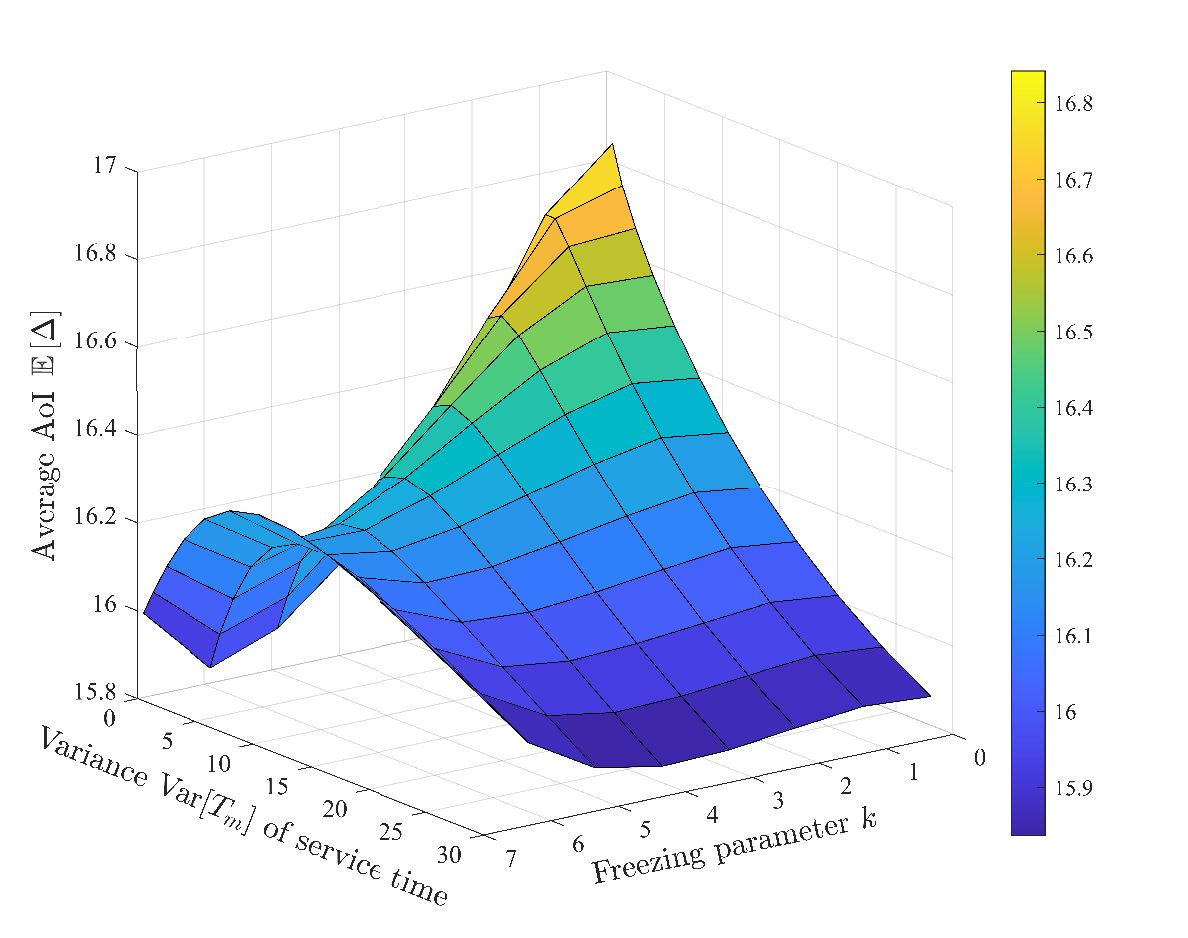}
\caption{Average AoI and the freezing parameter $k^*$ vs. the variance of service time for the triangular distribution. The service times $T_1$ and $T_2$ are identically distributed with $\mathbb E \left[T_m\right] = 13$, $m=1,2$.}
\label{AverageAoI_Var}
\end{figure}
Fig.~\ref{AverageAoI_Var} and Fig.~\ref{OptiFreezing_Var} demonstrate the impact of the variance of the service time on the average AoI and the optimum freezing parameter $k^*$, respectively.
We use triangularly distributed service times, the variance of which can be independently chosen from its mean.
It can be seen from Fig.~\ref{AverageAoI_Var} that for a given freezing parameter, the average AoI increases first and then decreases with increasing Var$[T_m]$. 
For the no-freezing case, a larger service time variance leads to a smaller average AoI.
On the other hand, we observe that with varying variance, the average AoI generally decreases first and then increases as the freezing parameter increases, again validating the benefit of freezing.
In particular, the use of freezing significantly reduces the average AoI when the variance is not very large, since the timeliness in the no-freezing case is quite poor at this point.
This fact is more clearly illustrated by Fig.~\ref{OptiFreezing_Var}, where the curves of percentage reduction in $\mathbb E\left[\Delta\right]$ decline overall (except for a negligible rise at the end) with the increase in $\text{Var}\left[T_m\right]$.
The maximum freezing gain, exceeding $17\%$, is obtained at Var$[T_m]=0$, where triangular distributions degenerate to deterministic distributions.
Specifically, the optimum freezing parameter for deterministically distributed service times is ${k^*} = {{\left\lfloor {\mathbb E\left[ {{T_m}} \right]} \right\rfloor } \mathord{\left/{\vphantom {{\left\lfloor {\mathbb E\left[ {{T_m}} \right]} \right\rfloor } 2}} \right.\kern-\nulldelimiterspace} 2}$, as it achieves the highest server utilization and the greatest diversity of packets across different servers simultaneously.
For $\text{Var}\left[T_m\right]>0$, the optimal freezing parameters decrease in a staircase pattern and can be obtained using the graph-based method.
In addition, we also observe that for a given Var$[T_m]$, the reduction in $\mathbb E\left[\Delta\right]$ increases when the mean service time increases, consistent with the cases utilizing geometrically and uniformly distributed service times.

\subsection{Examples with Non-Identical Distributed Service Times}\label{NonIden}

\begin{figure}
\centering
\includegraphics[width=0.47\textwidth]{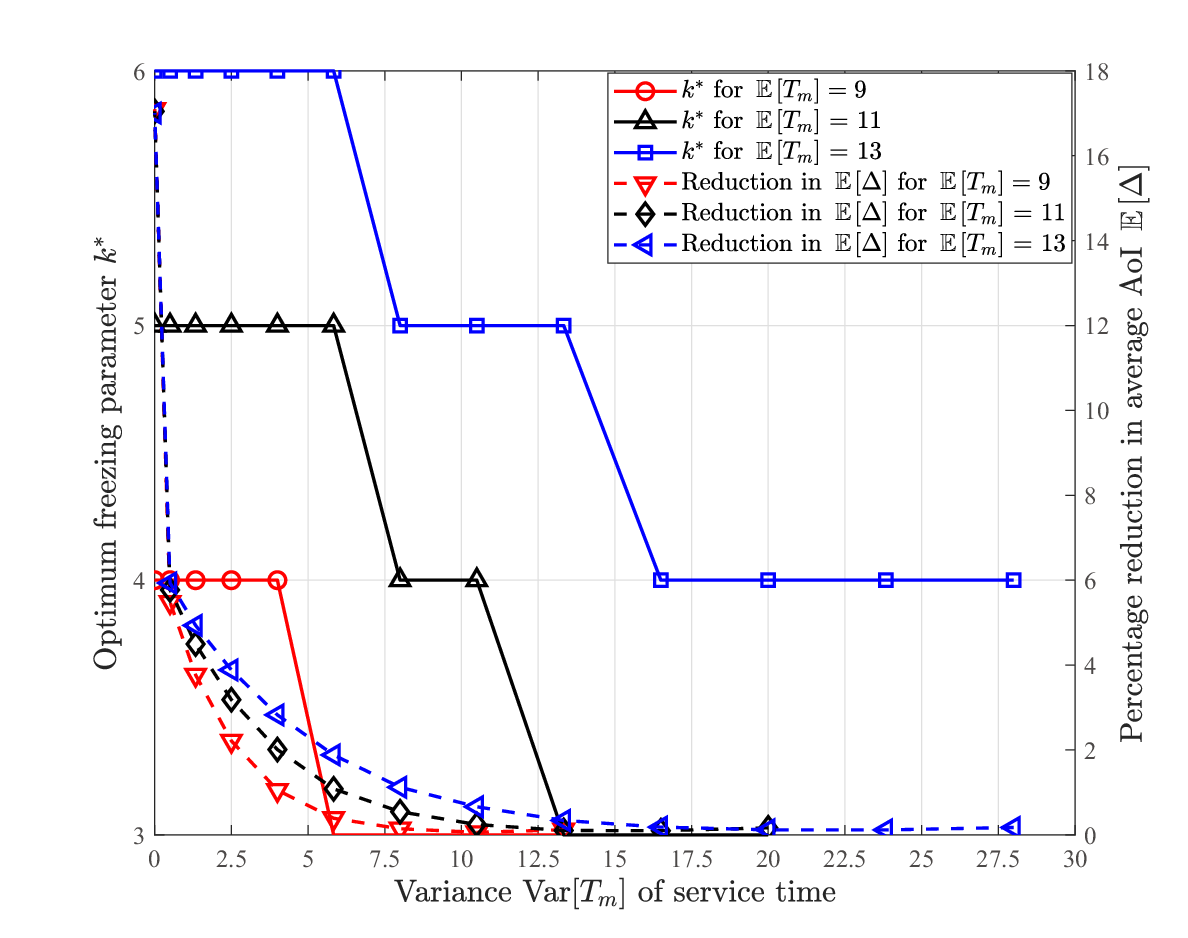}
\caption{Optimum freezing parameter and corresponding percentage reduction in average AoI versus variance of service time for triangular distribution. The service times $T_1$ and $T_2$ are identically distributed.}
\label{OptiFreezing_Var}
\end{figure}
All of the above numerical examples are obtained under the condition that the service times of $S_1$ and $S_2$ are identically distributed.
In Fig.~\ref{NonIden_Geo} and Fig.~\ref{NonIden_Tri}, we show the reduction in average AoI versus non-identical mean service times by using geometric distributions and triangular distributions, respectively.
Similar to Fig.~\ref{AverageAoI_k}(c) and Fig.~\ref{OptiFreezing_Var}, each point on the surface is obtained based on the corresponding optimal freezing parameter.
For geometrically distributed service times, the freezing gain shown in Fig.~\ref{NonIden_Geo} reaches $4.7\%$.
In addition, it can be seen that the effects of $\mathbb E \left[T_1\right]$ and $\mathbb E \left[T_2\right]$ on the reduction in $\mathbb E \left[\Delta\right]$ are distinguished.
Specifically, if $\mathbb E \left[T_1\right]$ is fixed, the reduction in $\mathbb E \left[\Delta\right]$ basically increases as $\mathbb E \left[T_2\right]$ increases.
If $\mathbb E \left[T_2\right]$ is fixed instead, the reduction in $\mathbb E \left[\Delta\right]$ increases first and then decreases with the increase in $\mathbb E \left[T_1\right]$.
The reason is that when both servers are idle, $S_1$ gives priority to serving a status update packet.
Using the policy of $S_2$ priority will change the shape of the surface.
For triangularly distributed service times, we fix the variances of the service times to $\text{Var}\left[T_1\right]$ = $\text{Var}\left[T_2\right] = 0.5$ in Fig.~\ref{NonIden_Tri}, while $\mathbb E \left[T_1\right]$ and $\mathbb E \left[T_2\right]$ are variable. 
It can be seen that the freezing gain reaches $6.17\%$.
Furthermore, the surface of the reduction in $\mathbb E \left[\Delta\right]$ presented in Fig.~\ref{NonIden_Tri} exhibits an interesting shape, indicating that the introduction of freezing only yields significant gains near the plane of $\mathbb E \left[T_1\right] = \mathbb E \left[T_2\right]$.
We observe that although the peak line satisfies $\mathbb E \left[T_1\right] = \mathbb E \left[T_2\right]$, which is different from Fig.~\ref{NonIden_Geo}, the surface is still asymmetrical due to the policy of $S_1$ priority.
Specifically, the surface is steeper on the side toward the $\mathbb E \left[T_1\right]$ axis than on the side toward the $\mathbb E \left[T_2\right]$ axis, indicating that the region with freezing gain is larger when $\mathbb E \left[T_1\right] < \mathbb E \left[T_2\right]$.
The asymmetry shown in Figs.~\ref{NonIden_Geo} and \ref{NonIden_Tri} suggest that, for cases with non-identical service times, conferring priority to the server with greater capacity can result in greater freeze gains, which is consistent with intuition.

\begin{figure}
\centering
\includegraphics[width=0.48\textwidth]{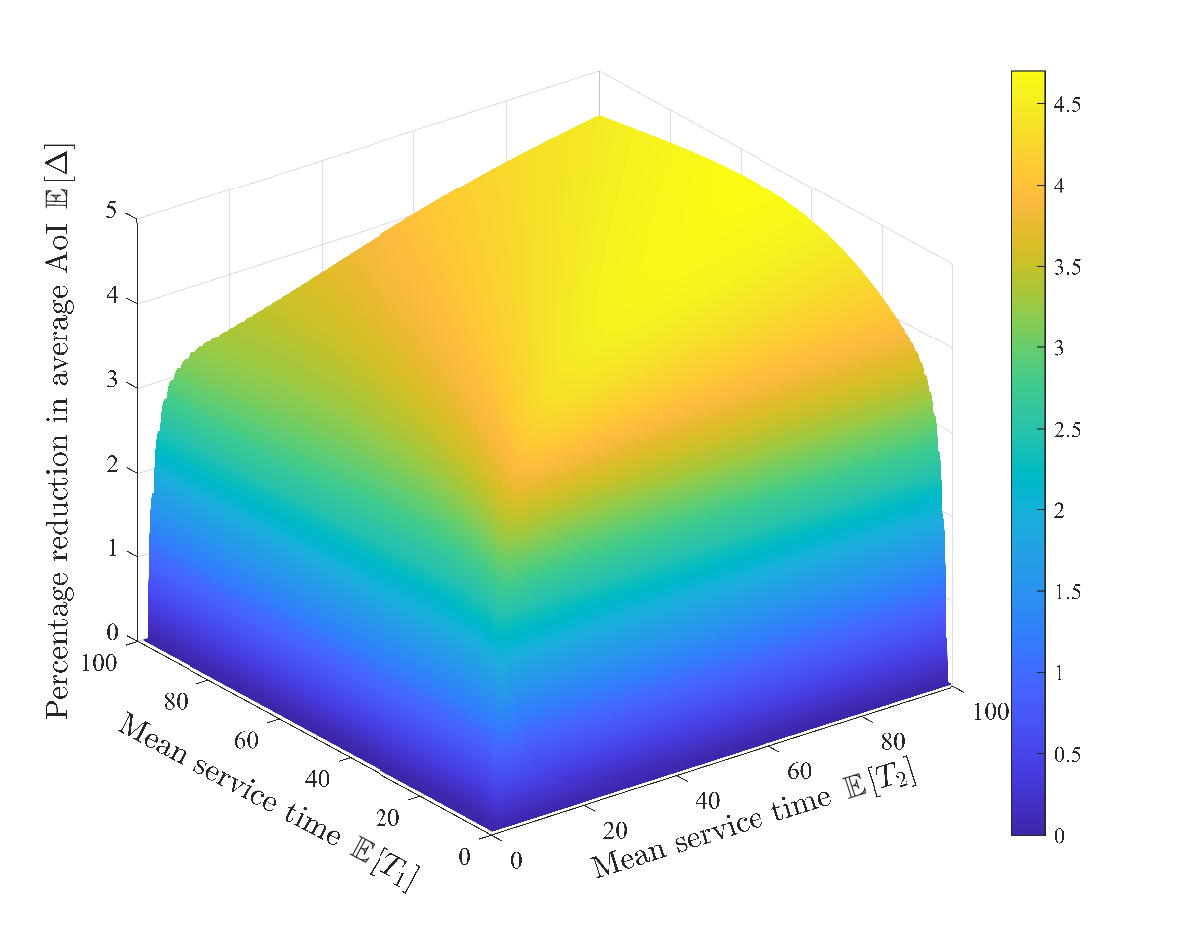}
\caption{Percentage reduction in average AoI versus mean service times of both servers for geometric distribution.}
\label{NonIden_Geo}
\end{figure}

\begin{figure}
\centering
\includegraphics[width=0.48\textwidth]{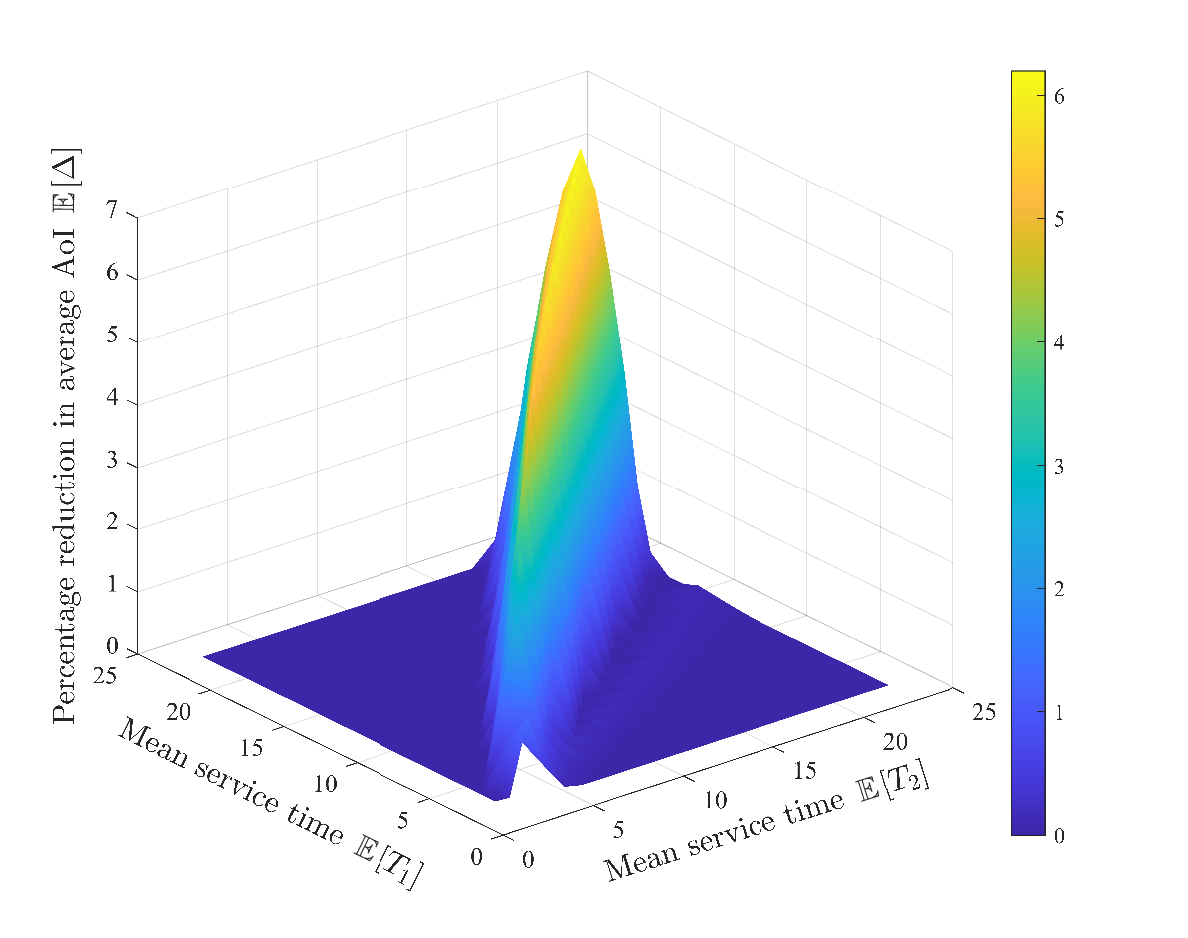}
\caption{Percentage reduction in average AoI versus mean service times of both servers for triangular distribution. $\text{Var}\left[T_1\right]$ = $\text{Var}\left[T_2\right] = 0.5$.}
\label{NonIden_Tri}
\end{figure}

\section{Conclusions}\label{Conclusion}
In this paper, we constructed a discrete-time AMC to model the distribution of steady-state AoI  of a DTDQ status update system considering transmission freezing.
DPH-type distributions were adopted to characterize the randomness of transmission times due to uncertain channel conditions, which can approximate arbitrary discrete distributions, enabling the universality of our model.
From the perspective of AMC, we further derived the exact distributions of both AoI and PAoI, thereby allowing us to easily obtain the mean of AoI and PAoI as well as their higher-order moments.
The validation of the constructed discrete-time AMC was conclusively demonstrated through several numerical examples using geometrically and uniformly distributed service times.
In addition, we investigated the role of freezing with various service time distributions.
The maximum freezing gain for geometric distributions appears when the mean service time of the priority server is shorter than the other server.
For triangular distributions, freezing is found to be more beneficial when the servers are identically distributed, and the gain reduces when the asymmetry increases.
In particular, using identical service times, freezing is especially beneficial for larger mean service times and smaller variances.
In this case, we have observed up to $6.17\%$ gains with freezing. 
Larger gains would have been possible but the model does not allow us to go to very high mean service times.
Future work can extend to more flexible server priority scenarios, e.g., depending on a random parameter, and investigation of the effect of this priority parameter on the timeliness and freezing gain for the DTDQ system. 
Study of DTDQ systems with random arrivals is also another interesting research direction.

\bibliographystyle{IEEEtran}
\bibliography{AoIDTDQ}

\begin{thebibliography}{10}
\providecommand{\url}[1]{#1}
\csname url@samestyle\endcsname
\providecommand{\newblock}{\relax}
\providecommand{\bibinfo}[2]{#2}
\providecommand{\BIBentrySTDinterwordspacing}{\spaceskip=0pt\relax}
\providecommand{\BIBentryALTinterwordstretchfactor}{4}
\providecommand{\BIBentryALTinterwordspacing}{\spaceskip=\fontdimen2\font plus
\BIBentryALTinterwordstretchfactor\fontdimen3\font minus \fontdimen4\font\relax}
\providecommand{\BIBforeignlanguage}[2]{{%
\expandafter\ifx\csname l@#1\endcsname\relax
\typeout{** WARNING: IEEEtran.bst: No hyphenation pattern has been}%
\typeout{** loaded for the language `#1'. Using the pattern for}%
\typeout{** the default language instead.}%
\else
\language=\csname l@#1\endcsname
\fi
#2}}
\providecommand{\BIBdecl}{\relax}
\BIBdecl

\bibitem{Survey1}
M.~Vaezi, A.~Azari, S.~R. Khosravirad, M.~Shirvanimoghaddam, M.~M. Azari, D.~Chasaki, and P.~Popovski, ``Cellular, wide-area, and non-terrestrial {IoT}: A survey on {5G} advances and the road toward {6G},'' \emph{{IEEE} Commun. Surveys Tuts.}, vol.~24, no.~2, pp. 1117--1174, 2022.

\bibitem{Survey2}
K.~Fizza, A.~Banerjee, P.~P. Jayaraman, N.~Auluck, R.~Ranjan, K.~Mitra, and D.~Georgakopoulos, ``A survey on evaluating the quality of autonomic {Internet of Things} applications,'' \emph{{IEEE} Commun. Surveys Tuts.}, vol.~25, no.~1, pp. 567--590, 2023.

\bibitem{Survey3}
D.~C. Nguyen, M.~Ding, P.~N. Pathirana, A.~Seneviratne, J.~Li, D.~Niyato, O.~Dobre, and H.~V. Poor, ``{6G} {Internet} of {Things}: A comprehensive survey,'' \emph{{IEEE} Internet Things J.}, vol.~9, no.~1, pp. 359--383, 2022.

\bibitem{Sensing}
S.~He, K.~Shi, C.~Liu, B.~Guo, J.~Chen, and Z.~Shi, ``Collaborative sensing in {Internet of Things}: A comprehensive survey,'' \emph{{IEEE} Commun. Surveys Tuts.}, vol.~24, no.~3, pp. 1435--1474, 2022.

\bibitem{AOI1}
M.~A. Abd-Elmagid, N.~Pappas, and H.~S. Dhillon, ``On the role of {Age of Information} in the {Internet of Things},'' \emph{{IEEE} Commun. Mag.}, vol.~57, no.~12, pp. 72--77, 2019.

\bibitem{OrignalAoI}
S.~Kaul, M.~Gruteser, V.~Rai, and J.~Kenney, ``Minimizing age of information in vehicular networks,'' in \emph{Proc. 8th Annu. IEEE Commun. Soc. Conf. Sensor Mesh Ad Hoc Commun. Netw.}, 2011, pp. 350--358.

\bibitem{PAOI1}
M.~Costa, M.~Codreanu, and A.~Ephremides, ``On the age of information in status update systems with packet management,'' \emph{{IEEE} Trans. Inf. Theory}, vol.~62, no.~4, pp. 1897--1910, 2016.

\bibitem{AOII}
A.~Maatouk, S.~Kriouile, M.~Assaad, and A.~Ephremides, ``The age of incorrect information: A new performance metric for status updates,'' \emph{{IEEE/ACM} Trans. Netw.}, vol.~28, no.~5, pp. 2215--2228, 2020.

\bibitem{VOI1}
Z.~Wang, M.-A. Badiu, and J.~P. Coon, ``A framework for characterizing the value of information in hidden {Markov} models,'' \emph{{IEEE} Trans. Inf. Theory}, vol.~68, no.~8, pp. 5203--5216, 2022.

\bibitem{VOI2}
T.~Soleymani, J.~S. Baras, and S.~Hirche, ``Value of {Information} in feedback control: Quantification,'' \emph{{IEEE} Trans. Autom. Control}, vol.~67, no.~7, pp. 3730--3737, 2022.

\bibitem{MSE1}
Y.~Sun, Y.~Polyanskiy, and E.~Uysal, ``Sampling of the {Wiener} process for remote estimation over a channel with random delay,'' \emph{{IEEE} Trans. Inf. Theory}, vol.~66, no.~2, pp. 1118--1135, 2020.

\bibitem{MSE2}
A.~Arafa, K.~Banawan, K.~G. Seddik, and H.~V. Poor, ``Sample, quantize, and encode: Timely estimation over noisy channels,'' \emph{{IEEE} Trans. Commun.}, vol.~69, no.~10, pp. 6485--6499, 2021.

\bibitem{MSE3}
Y.~Feng, Z.~Chen, M.~Motani, H.~H. Yang, M.~Wang, and T.~Q.~S. Quek, ``Distortion optimization for remote online estimation of the {Wiener} process,'' \emph{{IEEE} Trans. Wireless Commun.}, pp. 1--15, 2025.

\bibitem{AgeIsSem1}
Z.~Lu, R.~Li, K.~Lu, X.~Chen, E.~Hossain, Z.~Zhao, and H.~Zhang, ``Semantics-empowered communications: A tutorial-cum-survey,'' \emph{{IEEE} Commun. Surveys Tuts.}, vol.~26, no.~1, pp. 41--79, 2024.

\bibitem{AgeIsSem2}
E.~Uysal, O.~Kaya, A.~Ephremides, J.~Gross, M.~Codreanu, P.~Popovski, M.~Assaad, G.~Liva, A.~Munari, B.~Soret, T.~Soleymani, and K.~H. Johansson, ``Semantic communications in networked systems: A data significance perspective,'' \emph{{IEEE} Netw.}, vol.~36, no.~4, pp. 233--240, 2022.

\bibitem{AgeIsSem3}
M.~Kountouris and N.~Pappas, ``Semantics-empowered communication for networked intelligent systems,'' \emph{{IEEE} Commun. Mag.}, vol.~59, no.~6, pp. 96--102, 2021.

\bibitem{FCFS}
S.~Kaul, R.~Yates, and M.~Gruteser, ``Real-time status: How often should one update?'' in \emph{Proc. IEEE INFOCOM}, 2012, pp. 2731--2735.

\bibitem{LCFS}
S.~K. Kaul, R.~D. Yates, and M.~Gruteser, ``Status updates through queues,'' in \emph{Proc. 46th Annu. Conf. Inf. Sci. Syst. (CISS)}, 2012, pp. 1--6.

\bibitem{Replace}
M.~Costa, M.~Codreanu, and A.~Ephremides, ``On the {Age of Information} in status update systems with packet management,'' \emph{{IEEE} Trans. Inf. Theory}, vol.~62, no.~4, pp. 1897--1910, 2016.

\bibitem{General1}
Y.~Inoue, H.~Masuyama, T.~Takine, and T.~Tanaka, ``A general formula for the stationary distribution of the {Age of Information} and its application to single-server queues,'' \emph{{IEEE} Trans. Inf. Theory}, vol.~65, no.~12, pp. 8305--8324, 2019.

\bibitem{General2}
A.~Soysal and S.~Ulukus, ``{Age of Information in G/G/1/1} systems: Age expressions, bounds, special cases, and optimization,'' \emph{{IEEE} Trans. Inf. Theory}, vol.~67, no.~11, pp. 7477--7489, 2021.

\bibitem{PH1}
N.~Akar, O.~Do{\u{g}}an, and E.~U. Atay, ``Finding the exact distribution of {(Peak) Age of Information} for queues of {PH/PH/1/1 and M/PH/1/2} type,'' \emph{{IEEE} Trans. Commun.}, vol.~68, no.~9, pp. 5661--5672, 2020.

\bibitem{neuts81}
M.~F. Neuts, \emph{Matrix-Geometric Solutions in Stochastic Models: An Algorithmic Approach}.\hskip 1em plus 0.5em minus 0.4em\relax Dover Publications, Inc., 1981.

\bibitem{ocinneide}
C.~A. O’Cinneide, ``Characterization of phase-type distributions,'' \emph{Communications in Statistics. Stochastic Models}, vol.~6, no.~1, pp. 1--57, 1990.

\bibitem{SHS1}
R.~D. Yates and S.~K. Kaul, ``The {Age of Information}: Real-time status updating by multiple sources,'' \emph{{IEEE} Trans. Inf. Theory}, vol.~65, no.~3, pp. 1807--1827, 2019.

\bibitem{SHS2}
M.~Moltafet, M.~Leinonen, and M.~Codreanu, ``Average {AoI} in multi-source systems with source-aware packet management,'' \emph{{IEEE} Trans. Commun.}, vol.~69, no.~2, pp. 1121--1133, 2021.

\bibitem{SHS3}
A.~Javani, M.~Zorgui, and Z.~Wang, ``{Age of Information} in multiple sensing,'' in \emph{Proc. 2019 IEEE Global Commun. Conf. (GLOBECOM)}, 2019, pp. 1--6.

\bibitem{Mmt1}
R.~D. Yates, ``The {Age of Information} in networks: Moments, distributions, and sampling,'' \emph{{IEEE} Trans. Inf. Theory}, vol.~66, no.~9, pp. 5712--5728, 2020.

\bibitem{Mmt2}
M.~Moltafet, M.~Leinonen, and M.~Codreanu, ``Moment generating function of {Age of Information} in multisource {M/G/1/1} queueing systems,'' \emph{{IEEE} Trans. Commun.}, vol.~70, no.~10, pp. 6503--6516, 2022.

\bibitem{Flow1}
Z.~Chen, D.~Deng, H.~H. Yang, N.~Pappas, L.~Hu, Y.~Jia, M.~Wang, and T.~Q.~S. Quek, ``Analysis of {Age of Information} in dual updating systems,'' \emph{{IEEE} Trans. Wireless Commun.}, vol.~22, no.~11, pp. 8003--8019, 2023.

\bibitem{Flow2}
Y.~Qu, Z.~Chen, N.~Pappas, C.~Tang, M.~Wang, and T.~Q.~S. Quek, ``Analysis of {Age of Information} for a discrete-time dual-queue system,'' \emph{{IEEE} Trans. Commun.}, pp. 1--1, 2025.

\bibitem{AMC1}
N.~Akar and E.~O. Gamgam, ``Distribution of {Age of Information} in status update systems with heterogeneous information sources: An absorbing {Markov} chain-based approach,'' \emph{{IEEE} Commun. Lett.}, vol.~27, no.~8, pp. 2024--2028, 2023.

\bibitem{NZW1}
Y.~Sun, E.~Uysal-Biyikoglu, R.~D. Yates, C.~E. Koksal, and N.~B. Shroff, ``Update or wait: How to keep your data fresh,'' \emph{{IEEE} Trans. Inf. Theory}, vol.~63, no.~11, pp. 7492--7508, 2017.

\bibitem{AMC2}
N.~Akar and S.~Ulukus, ``Age of {Information} in a single-source generate-at-will dual-server status update system,'' \emph{{IEEE} Trans. Commun.}, pp. 1--1, 2025.

\bibitem{alfa_book}
A.~S. Alfa, \emph{Applied Discrete-Time Queues}.\hskip 1em plus 0.5em minus 0.4em\relax New York, NY, USA: Springer, 2016.

\bibitem{nielsen_book}
M.~Bladt and B.~F. Nielsen, \emph{Matrix-Exponential Distributions in Applied Probability}.\hskip 1em plus 0.5em minus 0.4em\relax New York, NY, USA: Springer, 2017.

\bibitem{telek_book}
L.~Lakatos, L.~Szeidl, and M.~Telek, \emph{Introduction to Queueing Systems with Telecommunication Applications}, 2nd~ed.\hskip 1em plus 0.5em minus 0.4em\relax New York, NY, USA: Springer, 2019.

\end{thebibliography}

\end{document}